\documentclass[a4paper]{article} 

\usepackage[english]{babel}
\usepackage{authblk}
\usepackage{dsfont}
\usepackage{amsfonts}
\usepackage{mathrsfs}
\usepackage{amssymb}        
\usepackage{amsmath}
\usepackage{graphicx}
\usepackage{float}
\usepackage[a4paper]{geometry} 
\usepackage{verbatim}
\usepackage{pstricks}
\usepackage{amsthm}
\usepackage{mathrsfs}

\DeclareMathOperator{\Tr}{Tr}

\newcommand{\Z}{{\mathbb Z}}
\newcommand{\R}{{\mathbb R}}

\newcommand{\ed}{{\mathrm e}}

\newtheorem{Prop}{Proposition}
\newtheorem{Theo}{Theorem}
\newtheorem{Def}{Definition}
\newtheorem{Remark}{Remark}

\newtheorem{Cor}{Corollary}

\begin{document}

\author[*]{Raphael Ducatez}
\author[*]{Fran\c cois Huveneers}
\affil[*]{{CEREMADE, Universit\' e Paris-Dauphine, France}}

\title{Anderson Localisation for periodically driven systems}                              

\maketitle

\begin{abstract}
We study the persistence of localization for a strongly disordered  tight-binding Anderson model on the lattice $\Z^d$, periodically driven on each site.
Under two different sets of conditions, we show that Anderson localization survives if the driving frequency is higher than some threshold value that we determine. 
We discuss the implication of our results for recent development in condensed matter physics, 
we compare them with the predictions issuing from adiabatic theory, 
and we comment on the connexion with Mott's law, derived within the linear response formalism. 
\end{abstract}

\tableofcontents 
\newpage
\section{Introduction}
In this paper, we study the fate of Anderson localization in periodically driven systems. 
Let $H_0$ be the tight-binding Anderson Hamiltonian on the lattice $\Z^d$. 
At strong enough disorder, it is well known that all eigenstates of $H_0$ are exponentially localized 
(see \cite{PhysRev.109.1492}\cite{frohlich1983}\cite{aizenman1993} as well as \cite{disertori2008random} for more references).
Let us then consider a periodic time-dependent Hamiltonian of the form 
\begin{equation}\label{general Hamiltonian}
H(t) = H_0 + g H_1 (t)
\end{equation}
with $H_1(t) = H_1(t+T)$ for some period $T$, and with $g$ some coupling constant. 
We assume that $H_1(t)$ acts everywhere locally: there exists $R$ such that $|(x, H_1(t) y )| = 0$ for all $x,y\in \Z^d$, and all time $t$, as soon as $|x-y| > R$
(with the notation $(x,Ay)=(\delta_x, A \delta_y) = A(x,y)$ for an operator $A$).

The time-evolution of an initial wave function $\psi (0)$ is governed by the time-dependent Schr\"odinger equation: 
$$ i \frac{d \phi (t)}{d t} = H(t) \phi (t).$$
The long time properties of the solutions of this equation are best understood through the Floquet eigenstates of $H(t)$ \cite{How}. 
The question addressed in this paper can then be rephrased as follows: 
Under suitable regularity conditions on the time-dependence of $H_1(t)$, is there a range of values for $g$ and $T$ such that the structure of the eigenfunctions of $H_0$ is only weakly affected by the periodic potential $H_1(t)$,
so that the the Floquet eigenstates of $H(t)$ are themselves localized?
We answer this question positively in Theorem \ref{Tlocalisation} below, for two different regularity conditions on $H_1$, leading to different allowed values for $g$ and $T$.

\paragraph{Localization and Floquet physics.}
The above question has already received some attention in the mathematical physics community. 
The connection with with the discrete non-linear Schr\"odinger equation (DNLS) constituted a first motivation, see \cite{Bourgain2004}\cite{Soffer-Wang}.
In this context, the more general case of a quasi-periodic driving shows up naturally: In a first approximation, the non-linearity in the DNLS equation can be replaced by a quasi-periodic perturbation. 
On the other hand, in this perspective, it is natural to restrict oneself to spatially localized perturbations 
($( x , H_1(t) y )$ decays fast as $x$ or $y$ goes to infinity and not only as $|x-y|$ goes to infinity as we consider); 
indeed, stability results for the DNLS equation all deal with originally localized wave packets. 

More recently, periodically driven Hamiltonian systems have been studied intensively in condensed matter theory. For two reasons at least:

First, from a theoretical perspective, driven systems constitute the first examples of dynamics out-of-equilibrium systems, lacking even energy conservation. 
The natural question that arrises is whether the system will absorb energy until it reaches an infinite temperature state (i.e.\@ a state with maximal entropy), as it would be the case for a chaotic system, 
or whether some extensive effectively conserved quantity emerges, forbidding energy absorption after some transient regime
\cite{Alessio-Polkovnikov}\cite{Alessio-Rigol}\cite{Abani-De_Roeck-Huveneers-PRL}\cite{Abanin-De_Roeck-Huveneer}\cite{Abanin-De_Roeck-Huveneers-Ho}\cite{Abanin-De_Roeck-Ho}.
For non-interacting particles on a lattice, as we consider in this paper, this issue becomes trivial and fully independent of the issue of Anderson localization, 
once the driving frequency becomes higher than the bandwidth of individual particles, see \cite{Abani-De_Roeck-Huveneers-PRL}\cite{Abanin-De_Roeck-Huveneers-Ho}.  
Nevertheless, thanks to the Anderson localization phenomenon, our results guarantee the existence of an effective extensive conserved quantity for frequencies much below this trivial threshold, 
see Proposition \ref{CUnitaire} below. 
 
Second, from a more practical point of view, driven systems furnish a way to engineer topological states of matter \cite{Oka-Aoki}\cite{Linder-Refael-Galitski}.  
Though this possibility is not apriori related to the phenomenon of Anderson localization, it turns out that, 
for interacting many-body systems, localization makes it possible to ``lift'' phase transitions from the ground state to the full spectrum \cite{Huse-Nandkishore}. 
This observation is at the heart of very recent investigations of new phases of matter inside the many-body localized phase \cite{Khemani-Lazarides-Moessner-Sondhi}\cite{von_Keyserlingk-Sondhi}\cite{von_Keyserlingk-Sondhi_2}.

Hence, in view of the increasing role played by localized Floquet systems in modern condensed matter physics, 
it appeared useful to bring some firm mathematical foundations to the theory of Anderson localization in periodically driven systems, 
even though the need for mathematical rigor forces us to restrict the setup to non-interacting particles. 
Results in this direction already appeared in \cite{hamza}, where the localization for some random unitary operators is established;  
this question is directly related to ours since the long time evolution of a periodically system is governed by the spectral properties of the unitary $U(T)$, where $U(t)$ solves $i d U(t)/dt = H(t) U(t)$. 
However, for a Hamiltonian as in \eqref{general Hamiltonian}, we do not recover the particular form for $U$ studied in \cite{hamza}.  

Before stating our results, we now introduce two more specific aspects that deserved clarification and motivated the present article. 

\paragraph{Adiabatic Theory.}
Time-dependent Hamiltonian systems varying smoothly and slowly enough with time can be described through the use of adiabatic theory.
Here, adapting the analysis from \cite{Abanin-De_Roeck-Huveneer},  we argue that localization emerges when level crossings in the system become typically non-adiabatic, and we determine the threshold frequency above which this happens. 

Let us first remind the theory of the Landau-Zener effect for a time-dependent two-levels Hamiltonian $G(t)$ \cite{Landau}\cite{Zener}. 
To make the connection with our problem, let us assume that $G(t)$ is of the form $G(t) = P H(t) P$ where $P$ projects on two eigenstates of $H_0$. 
Moreover, we assume that $G(t)$ varies smoothly on the scale of one period, i.e.\@ we can write $G(t) = \tilde G (\nu t)$ for some smooth $2\pi$-periodic function $\tilde G$ and $\nu = \frac{2\pi}{T}$. 
It is then convenient to move to the basis of the eigenstates of $H_0$, i.e.\@ the basis where $P$ is diagonal, and to decompose  
$$G(t) = G_{\mathrm{dia}}(t) + G_{\mathrm{off}}(t),$$
as a sum of the diagonal and off-diagonal part.
We notice that the time-dependent part of $G_{\mathrm{dia}}(t)$ is of order $g$. 
We set $( 1,G_{\mathrm{off}}(t)2) =: g'$, where $g'$ depends mainly on the distance between the two localization centers of the two states projected on by $P$, and is typically much smaller than $g$. 
Finally we assume that the two levels of $PH_0P$ are close enough ($g$-close in fact) to each others so that the system undergoes an avoided crossing as time evolves: 
At some time, the levels of $G_{\mathrm{dia}}(t)$ cross, while $G_{\mathrm{off}}(t)$ induces level repulsion, leading to an avoided crossing for $G(t)$. 

If the system is initially (i.e.\@ before the crossing) prepared in an eigenstate of $G_{dia}(t)$, Landau-Zener theory tells us that, after the crossing, the state in which the systems ends up depends on the value of 
\begin{equation}\label{Landau Zener}
\frac{|( 1,G_{\mathrm{off}}2) |^2}{v_{12}} \sim \frac{(g')^2}{g \nu}, 
\end{equation}
where $v_{12}$ is the rate of change in the energy  of $H_{dia}(t)$ at the crossing. 
At high frequency, when this value is much smaller than $1$, the crossing is non-adiabatic and the system remains in the original state; 
at intermediate frequency, when this value is of order $1$, the system ends up in a superposition of the eigenstates of $G_{\mathrm{dia}}(t)$; 
and finally at low frequency, when this value is much smaller than $1$, the crossing is adiabatic and the system ends up in the other eigenstate of $G_{\mathrm{dia}}(t)$.   

The above scenario, valid for a two level systems, may be seen as a caricature of the localization-delocalization transition: 
non-adiabatic crossings do not entail hybridization of unperturbed eigenstates, 
while intermediate and adiabatic crossings, present at low enough frequency, allow the system to move from one state to the other, and constitute a possible mechanism for delocalization.  
Based on this picture, let us try to determine a critical value of $\nu$ above which localization survives.  
Let us fix $g$ in \eqref{general Hamiltonian} as well as $W$ characterizing the strength of the disorder. 
Let us then pick a point $a\in \Z^d$. 
We first determine a minimal length $L^*$ so that there is typically at least one crossing between the state centered around $a$ and an other state with localization center in a ball of radius $L^*$ around $a$. 
Since the probability of finding a crossing in a ball of radius $L$ is of the order of $L^d \frac{g}{W}$, we find 
$$L^* \sim   \Big( \frac{W}{g} \Big)^{1/d}.$$
The effective coupling between a state centered around $a$ and a state at a distance $L$ of $a$, corresponding to $g'$ in \eqref{Landau Zener} is of the order of 
$$g' \sim g \ed^{-L/\xi},$$
where $\xi$ is the localization length of $H_0$. Hence, from \eqref{Landau Zener}, we find that localization will survive if 
\begin{equation}\label{condition frequency from LZ}
\frac{ g^2 \ed^{-2L^*/\xi}}{g \nu} \ll 1 \qquad \Leftrightarrow  \qquad \frac{\nu}{g} \gg \ed^{-\frac{2}{\xi}(\frac{W}{g})^{1/d}}. 
\end{equation}

In Theorem \ref{Tlocalisation} below, for a smooth driving (condition (C1)), we prove localization for $\nu$ larger than some threshold value comparable to what we obtain in \eqref{condition frequency from LZ}. 
We notice that the Landau Zener theory proceeds through non-perturbative arguments. 
Instead, our proof is based on the multi-scale analysis developed in \cite{frohlich1983}, which is mainly a perturbative approach. 
It is thus somehow remarkable that the same upshot can be recovered in two a priori very different ways. 

Finally we notice that the approach through adiabatic theory outlined above is only expected to work for $H_1(t)$ depending smoothly on time. 
Unfortunately, both in theoretical and experimental physics works, it is a common protocol to just shift between two Hamiltonians periodically. 
This leads obviously to a non-smooth time-dependence. As we wanted to cover this case as well, we also derived a result for $H_1$ being only in square-integrable in time; 
see Theorem \ref{Tlocalisation} below with the condition (C2).
The lack of smoothness forced us to increase significantly the threshold on $\nu$ with respect to \eqref{condition frequency from LZ}.

\paragraph{Mott's law.} 
Mott's law asserts that the ac-conductvity of an Anderson insulator behaves as 
$$\sigma (\nu) \; \sim \; \nu^2 \big( \log (1/\nu) \big)^{d+1} \quad \text{as} \quad \nu \to 0$$ 
(\cite{Mott}, see also \cite{Gopalakrishnan_Mueller} for the case of interacting electrons). 
An upper bound on $\sigma (\nu)$ was rigorously established in  \cite{Klein-Lenoble} (with $d+1$ replaced by $d+2$). 
The conductivity $\sigma (\nu)$ is derived within the linear response (LR) formalism;
in our set-up, this corresponds to fixing $\nu$ and taking the limit $g\to 0$ while observing the dynamics over a time of order $\nu/g^2$. 
In such a regime, the hypotheses of Theorem \ref{Tlocalisation} below are satisfied (we consider a monochromatic perturbation with frequency $\nu$ so that condition (C1) holds): 
The dynamics is localized for $g$ small enough once $\nu$ has been fixed.\footnote{
Strictly speaking, our model does not coincide with that studied in e.g.\@ \cite{Klein-Lenoble}, as we do not explicitly include an electric field. However, it could be incorporated without affecting our conclusions.}
It may thus come as a surprise that still $\sigma (\nu) > 0$. 

This puzzling behavior was recently analyzed in details for many-body systems in \cite{Gopalakrishnan_Knap}. 
As it was pointed out to us by \cite{De_Roeck}, the conductivity $\sigma (\nu)$ is computed for a system in equilibrium at zero or finite temperature. 
Moreover, as can be expected from its definition, for $g>0$, LR should in general furnish only an accurate description of the dynamics for a transient regime in time of order $\nu/g^2$. 
It is true though that, for ``generic"  or ``ergodic" systems, it is reasonable to think that the predictions from LR remain valid for much longer time scales: 
While heating, the system remains approximately in equilibrium and LR can be applied iteratively until the infinite temperature state is reached. 
This is manifestly not true for localized systems as long as $g$ is small enough compared to $\nu$:
The conductivity $\sigma (\nu ) > 0$ represents mainly the Rabi oscillation of rare resonant spots (``cat states") in the Hamiltonian $H_0$, 
but  these oscillations do not need to entail delocalization on the longes time scales described by the Floquet physics.

\paragraph{Organisation of the paper.}
The precise definition of the model studied in this paper together with our results are presented in Section \ref{Model and results}. 
The main steps of the proof of our main theorem are contained in Section \ref{outline proof Theorem}, while some more technical intermediate results are shown in Sections \ref{SWegner} to \ref{section: general L2 case}. 
The two corollaries are shown in Section \ref{section: dynamical localization}. 
In several places, the proof of our results proceeds through a straightforward adaptation of delicate but well-known methods; 
as much as possible, we choose to describe in details only the steps where some significant amount of new material was required.

\paragraph{Acknowledgments.} 
We are especially grateful to W.~De Roeck for enlightening discussions on Mott's law as well as previous collaborations on this topic. 
We thank D.~Abanin, W.-W.~Ho and M.~Knap for previous collaboration and/or useful discussions.

\section{Models and results}\label{Model and results}

\subsection{The models}

We consider a lattice model on $\mathbb{Z}^d$ and we note $|x|=\sup_{i=1..d}|x_i|$. Our results could be of course extended to more general lattices.  
We are interested in the long time behavior of the Schr\"odinger equation: 
\begin{equation} \label{Ealpha}
i\frac{d}{dt}\phi (t)= H(t) \phi (t),
\end{equation}
where the function $\phi(t)$ is defined on $L^2(\mathbb{Z}^d)$ for any $t$, and the Hamiltonian $H(t)$ is a periodic function with frequency $\nu=2\pi/T$. 
The operator $H(t)$ is an idealized version of \eqref{general Hamiltonian}: 
We move to the basis where $H_0$ is diagonal and we replace it by an uncorrelated random potential $V_\omega$, while we assume that $H_1 (t)$ is still a nearest-neighbor hopping (Anderson model): 
\begin{equation}
H(t) = -g \Delta(t) + V_\omega.
\end{equation}

Here $-\Delta(t)$ is hermitian operator for any $t$ such that $-\Delta(t)(x,y)=0$ if $|x-y|>1$ and
\begin{align}\label{LtwoBound}
& \|-\Delta(t)(x,y)\|_{{L}^2([0;T])}\leq 1
\end{align}
for any $x$, $y$.
We use the notation $-\Delta$ because in the usual time-independent Anderson model, $-\Delta(t)$ is the usual discrete Laplacian on $\ell(\mathbb{Z}^d)$
\begin{equation*}
-\Delta \phi (x) =\frac{1}{2d} \sum_{|y-x|=1} \phi(y),
\end{equation*}
There exists a unitary operator $U(t)$, with $U(0)=Id$ such that $\phi(t)=U(t)\phi(0)$ and satisfying 
\begin{equation} \label{Eunitary}
i\frac{d}{dt}U (t)= H(t) U (t),
\end{equation}
Existence and uniqueness of solution of \eqref{Ealpha} and \eqref{Eunitary} can be proved using a usual fixed point technique.

\bigskip
\textbf{(RP) Potential regularity.} We assume the following form for the random potential which are widely used in the literature:
\begin{equation}
V_\omega = \sum_{x\in\mathbb{Z}^d} v_x \delta _x
\end{equation} where $v_x$ are i.i.d.\@ random variables, with a bounded density $\rho$, such that $\|\rho\|_\infty < \infty $ defined on a bounded support $[-M;M]$. 
We choose units such that $\| \rho \|_\infty = 1$. 
Furthermore we will assume that the density $\rho$ is piecewise $\mathcal{C}^1$.

\bigskip 
The time-dependent term $-g\Delta (t)$ is considered to be a perturbation of order $g\ll 1$, usually referred to as the strong disorder regime. We treat this model in two particular cases. 

\bigskip
\textbf{(C1) Smooth driving.} We suppose that  $-\Delta(t)(x,y)$ is a monochromatic signal: For any $x$ and $y$,
\begin{align}
& -\Delta(t)(x,y)=a_{x,y} + b_{x,y} \cos(\nu t)+b_{x,y}'\sin(\nu t)
\end{align}
with $a_{x,y}=a_{y,x}$, $b_{x,y}'=b_{y,x}'$ and $b_{x,y}=b_{y,x}$. 
In this regime, we are able to prove localization for frequencies $\nu$ up to a threshold comparable to the one given in \eqref{condition frequency from LZ}. 
Moreover, we claim that the result can then be extended to a hopping $-\Delta (t)$ with Fourier coefficients that decay fast enough, but we focus on the case of single Fourier mode for simplicity.   

\bigskip
\textbf{(C2) $L^2$ driving.} 
We only assume  \eqref{LtwoBound}. In this case, a much larger threshold value for $\nu$ is needed, actually $\nu \ge 1$. 
We refer to \cite{Abanin-De_Roeck-Huveneer} for the optimality of this condition. 

\begin{Remark}
Between these two extreme cases, one could obviously consider intermediate regularity cases, depending on the decay of the Fourier coefficients of $-\Delta (t)$. 
This should lead to other conditions on $\nu$ that are not investigated in this paper. 
\end{Remark}

\subsection{The Floquet operator}

We will work in the Fourier space instead of the time-domain, and we denote by $\hat{x} = (x,k)$ a point of $\Z^d \times \Z$. 
Let's introduce the central object of our paper:
\begin{Def} Let
\begin{equation}
\hat{H}= -g \hat{\Delta} + \hat{V}_\omega 
\end{equation}
be a Hamiltonian on $\mathbb{Z}^d \times \mathbb{Z}$, with 
\begin{equation}
-\hat{\Delta}\hat{\psi}(x,k) = -\sum_{|y-x|\leq 1} \sum_{k'}\hat{\Delta}_{x,y}(k')\hat{\psi}(y,k-k')
\end{equation}
where $\hat{\Delta}_{x,y}(k) = \frac{1}{T}\int_0^T \Delta_{x,y} (t) e^{-i \nu k t} dt$ and
\begin{equation}
\hat{V}_\omega = V_\omega + k\nu .
\end{equation}
\end{Def}

In the mono-chromatic case (C1), the Laplacian $-\hat{\Delta}$ is explicitly given by   
\begin{align*}
&-\hat{\Delta}\hat{\psi}(x,k) = \sum_{|y-x|\leq 1} \Big[a_{x,y} \hat{\psi}(y,k) + \frac{b_{x,y}+i b_{x,y}' }{2}\hat{\psi}(y,k+1)+\frac{b_{x,y}-i b_{x,y}' }{2}\hat{\psi}(y,k-1) \big)\Big] 
\end{align*}
We remark that it is a local operator, meaning it connect only sites $\hat{x}, \hat{y}$ such that $|\hat{x}-\hat{y}|= 1$ in the space-Fourier graph $\mathbb{Z}^d\times\mathbb{Z}$. In the general $L^2$ case (C2), this is no longer true. Indeed, points $(x,k),(y,k')$ could be connected with $|k-k'|$ arbitrary large.

The new Hamiltonien $\hat{H}$ gives the evolution of the ``finite time Fourier series" of $\phi(t)$ defined as follows
\begin{equation}
\check{\phi}(x,k,t)=\frac{1}{T}\int_t^{t+T}\phi(x,u)e^{-i\nu k u}du.
\end{equation}
We get formally a time-independent Schr\"odinger equation governed by the Hamiltonian $\hat{H}$:

\begin{Prop}
\begin{equation}
i\partial_t \check{\phi}(x,k,t)=\hat{H}\check{\phi}(x,k,t)
\end{equation}
\end{Prop}
\begin{proof}
\begin{align*}
& i\partial_t \check{\phi}(x,k,t) = \frac{1}{T}\int_t^{t+T}i\partial_u\big[\phi(x,u)e^{-i\nu k u}\big]du \\
& \qquad = \frac{1}{T}\int_t^{t+T} \big( k\nu+H(u)\big) \phi(x,u)e^{-i\nu k u}du \\
& \qquad = \frac{1}{T}\int_t^{t+T} (k\nu+V_\omega)\phi(x,u)e^{-i\nu k}+ g \sum_{|y-x|\leq 1}\sum_{k'}(-\hat{\Delta}_{x,y}(k')) \phi(y,u)e^{-i\nu (k-k') u}du \\
& \qquad =(V_\omega +k\nu)\check{\phi}(x,k,t)+ g \sum_{|y-x|\leq 1} \sum_{k'} (- \hat{\Delta}_{x,y}(k'))\check{\phi}(y,k-k',t)
\\& \qquad =  \hat{H}\check{\phi}(x,k,t).
\end{align*}
\end{proof}

The time evolution of $\check{\phi}$ is deduced from the eigenvectors of $\hat{H}$:
\begin{equation}\label{EFouEig}
\bar\lambda \hat{\psi} =\big( -g \hat{\Delta} + \hat{V}_\omega \big) \hat{\psi} 
\end{equation}
Looking for the eigenvectors of $\hat{H}$ is equivalent to the search of  solution of the form $\phi(t) = e^{i\bar\lambda t}\psi(t)$ with $\psi$ a $T$-periodic function (Floquet theory).
Indeed,  in the Fourier variables, \eqref{Ealpha} is equivalent to \eqref{EFouEig}.
In particular, as we will see, localization for $\hat{H}$ implies the absence of diffusion for $\phi$. 


\begin{Remark}\label{RSymmetry} 
Because $\psi(t)e^{i\bar\lambda t} = \psi(t)e^{-i n \nu t} e^{i (n \nu+\bar\lambda) t} $, if $(\psi,\bar\lambda)$ is a solution then ($\psi(t)e^{-i n \nu t}, n\nu + \bar\lambda)$ is a solution as well for any $n\in \mathbb{Z}$. 
Hence it is enough to consider the case $\bar\lambda\in [0;\nu]$.
\end{Remark}

\subsection{Results}

Our main theorem states Anderson localisation for $\hat{H}$. 

\begin{Theo}\label{Tlocalisation}
There exists $\epsilon>0$ such that, if $g<\epsilon$, and if $\nu\geq  e^{-g^{-\frac{1}{4p+8d}}}$ for some $p>2d$ under the condition (C1), or if $\nu\geq 1$ under the condition (C2),
then $\hat{H}$ exhibits localization :Its spectrum is pure point and its eigenvectors decay exponentially in space, $\mathbb{P}$ a.s. 
\end{Theo}

\begin{Remark}
Under (C1), we will see that the eigenvectors are also deterministically exponentially localized along the frequency axis. 
\end{Remark}

The two following corollaries do not logically follow from Theorem \ref{Tlocalisation}, but rather from a refinement of its proof. 
The first one shows the absence of diffusion for solutions of \eqref{Ealpha} (dynamical localization): 

\begin{Cor} \label{TAbsenceDiff}
There exist $\epsilon>0$ and $q>0$ (and one may take $q\to \infty$ as $\epsilon \to 0$) such that, if $g<\epsilon$ and $\nu\geq  e^{-g^{-\frac{1}{4p+8d}}}$ for some $p> 2d$ under (C1), or $\nu\leq 1$ under (C2), then 
\begin{equation}\label{AbsenceDiff}
\mathbb{E}\Big(\sup_{t>0}  \sum_{x\in \Z^d}|x|^q |\phi(x,t)|^2 \Big) < \infty 
\end{equation}
for any solution $\phi(x,t)$ of \eqref{Ealpha} with initial condition $\phi(x,0)$ defined on a bounded support.   
\end{Cor}

The second one deals with the existence of a local effective Hamiltonian, i.e.\@ an Hamiltonian $H_{eff}$ such that 
$$U(T) = e^{- i T H_{eff}} $$
and such that $H_{eff} (x,y)$ decays fast as $|x-y| \to \infty$. 
Under the conditions of  Theorem \ref{Tlocalisation}, given $\bar\lambda \in [0,\nu[$ and a corresponding eigenfunction $\hat{\psi}_{\bar\lambda}(k,x)$ of $\hat{H}$, and given $t\in \R$, 
let us denote by $P_{\psi_{\bar\lambda}(\cdot,t)}$ the projector 
$$ L^2 (\Z^d) \to L^2 (\Z^d), f \mapsto \big(\psi_{\bar\lambda}(\cdot , 0)  , f \big) \, \psi_{\bar\lambda}(\cdot , t).$$
The representation 
$$U(t) = \sum_{\bar\lambda \in [0,\nu[} e^{-i\bar\lambda t} P_{\psi_{\bar\lambda}(\cdot,t)}$$  
holds. Hence, since the functions $\psi_{\bar\lambda}(\cdot ,t)$ are $T$-periodic in time, we may set
\begin{equation}\label{H eff}
H_{eff} =  \sum_{\bar\lambda \in [0,\nu[} \bar\lambda P_{\psi_{\bar\lambda}(\cdot,0)}, 
\end{equation}
which defines an operator on $L^2(\Z^d)$. Under condition (C1), we have a more\footnote{
The result would be of little interest under condition (C2), since at high frequency, the existence of a local effective Hamiltonian follows from much more general considerations, see \cite{Abanin-De_Roeck-Huveneers-Ho}.}
\begin{Cor}\label{CUnitaire}
There exist $\epsilon>0$ and $q>0$ (and one may take $q\to \infty$ as $\epsilon \to 0$) such that, if $g<\epsilon$,  $\nu\geq  e^{-g^{-\frac{1}{4p+8d}}}$ for some $p>2d$, and under condition (C1), then 
\begin{equation*}
\mathbb{E} \big( |x-y|^q | H_{eff} (x,y) | \big) < \infty.  
\end{equation*}
with $H_{eff}$ as defined by \eqref{H eff}. 
\end{Cor}

\section{Proof of Theorem \ref{Tlocalisation}}\label{outline proof Theorem}

We will prove that the Hamiltonian $\hat{H}$ reveals localisation by applying the classical tools of the multi-scale analysis (MSA). 
Thanks to the huge literature on MSA, it we will be enough for us to prove a probability estimate, usually referred to as Wegner estimate, and the initialization of the MSA to show the localisation 
(as well as some extra technical results when dealing with the  $L^2$ case, i.e.\@ under assumption (C2)). 

We start with the Wegner estimate. 
Below we call columns sets of the form $\Lambda_0 \times I \subset \Z^d \times \Z$, for some finite spatial box $\Lambda_0$ and some frequency interval $I$. 
Given $\Lambda \subset \Z^d \times \Z$ and given $H\in L^2 (\Z^d \times \Z)$, 
we denote by $H_{|\Lambda}$ the operator acting on $L^2(\Lambda)$ such that $H_{|\Lambda}(\hat{x},\hat{y}) = H(\hat{x},\hat{y})$ for all $\hat{x},\hat{y} \in \Lambda$. 
\begin{Prop}[Wegner Estimate]\label{PWegner}
Let $\Lambda_0 \subset \mathbb{Z}^d$ be finite. Then
\begin{enumerate}
\item (The finite column case) For any $K\in \mathbb{N}$, $k_0\in \mathbb{Z}$ so that $\Lambda_0\times[k_0-K;k_0+K]\subset\mathbb{Z}^d\times\mathbb{Z}$, we have
\begin{equation}\label{EWegnerFinite}
\forall E ,  \mathbb{P}(\exists \bar\lambda \mbox{ eigenvalue of $\hat{H}_{|\Lambda_0\times[k_0-K;k_0+K]}$ : } \bar\lambda \in [E-\epsilon , E+\epsilon ])\leq 2\pi \epsilon (2K+1)|\Lambda_0| ||\rho||_\infty.
\end{equation}
\item (The infinite column case) There exists a constant $C$ which depends only on $\|\rho\|_{L^\infty}$ and $\|\rho'\|_{L^\infty}$, such that for $\Lambda_0\times\mathbb{Z}\subset\mathbb{Z}^d\times\mathbb{Z}$, we also have
\begin{equation}\label{EWegnerInfinite}
\mathbb{P}(\exists \bar\lambda \mbox{ eigenvalue of $\hat{H}_{|\Lambda_0\times\mathbb{Z}}$ : } \bar\lambda \in [E-\epsilon , E+\epsilon ])
 \leq 2\pi \sqrt{\epsilon} |\Lambda_0| ||\rho||_\infty \max(1,\frac{M}{\nu}).
\end{equation}
\end{enumerate}
\end{Prop}

The proof of this proposition will be carried over in section \ref{SWegner}.
Part 1.\@ will be needed to establish Theorem \ref{Tlocalisation} under the assumption (C1) and part 2.\@ under the assumption (C2). 
The crucial property that allows to show the second part of this proposition is contained in Remark \ref{RSymmetry}: 
If $\hat\psi(x,k)$ is an eigenvector with eigenvalue $\bar\lambda$ of $\hat{H}_{|\Lambda_0\times\mathbb{Z}}$, then $\hat\psi(x,k-k_0)$ is also an  eigenvector with eigenvalue $\bar\lambda+\nu k_0$ for any $k_0 \in \Z$. 
Therefore the eigenvalue are of the form $\{\bar\lambda_i : i=1,\dots,|\Lambda_0|\}+\nu\mathbb{Z}$, allowing to use $|\Lambda_0|$ in the rhs of \eqref{EWegnerInfinite} instead of the cardinal of the column which in this case is infinite.


The second ingredient in the MSA consists in proving the exponential decay of the resolvent $(\hat{H}- \lambda)^{-1}$ with high probability for a given $\lambda \in \R$. 
We will follow \cite{disertori2008random}. 
To initialize the MSA, we need to show that, given a point $\hat{x}\in \Z^d \times \Z$, there exists with high probability a finite domain around $\hat{x}$, called ``good box", where the resolvent decay exponentially. 
From now on we fix some $\lambda\in [0,\nu]$.
Indeed, it is enough to consider values of $\lambda$ in this interval, because of the symmetry described in Remark \ref{RSymmetry}. 

For $\Lambda\subset \mathbb{Z}^d\times\mathbb{Z}$, we will write
\begin{equation}
\partial^{in} \Lambda=\{\hat{x}\in\Lambda : \exists\hat{y}\notin\Lambda, \hat{\Delta}(\hat{x},\hat{y})\neq 0\}
\end{equation}
\begin{equation}
\partial^{ext} \Lambda=\{\hat{x}\notin\Lambda : \exists\hat{y}\in\Lambda ,\hat{\Delta}(\hat{x},\hat{y})\neq 0\}
\end{equation}

\subsection{Smooth driving (C1)}

\begin{Def}[Good box]
Under the assumption (C1), we say that $(x+[-L,L]^d)\times[k_0-K,k_0+K]$ is a $\mu$-good box, for some $\mu >0$, if, for any $(y,k)\in\partial^{in} \Big(x+[-L,L]^d)\times[k_0-K,k_0+K]\Big)$,
\begin{equation} \label{EgoodC}
 |\big((x,k_0),\big(\hat{H}_{|(x+[-L,L]^d)\times[k_0-K,k_0+K]}-\lambda \big)^{-1} (y,k)) |\leq e^{-\mu |(x,k_1)-(y,k_2)|}
 \end{equation}
 where $|(x,k_1)-(y,k_2)|=|k_2-k_1|+\sum_{i=1}^d|x_i-y_i|$ .
 \end{Def}

The difference between our model and the classical Anderson model is the absence of independence along the frequency axis. However we have the following proposition.

\begin{Prop} \label{PlargeK} If $|k_0|>\frac{M+\sqrt{g}}{\nu}+K$ then for any $\Lambda_0\subset\mathbb{Z}^d$, $\Lambda_0\times [k_0-K;k_0+K]$ is a $-\ln(2(d+1)g)$ good box.
\end{Prop}

The proof of this proposition will appear as a simple case of the proof of Proposition \ref{InitiationMulti} below (see Section \ref{SHarmonic} after the proof of Proposition \ref{DecayFromChain}). 
Thanks to this proposition, it is now enough then to study boxes close to the $k=0$ axis. 
Once we restrict ourselves to such boxes, non-intersecting boxes are stochastically independent, and we can proceed with the usual MSA approach.   
So the idea of the proof is to show initialization of the MSA for boxes like $\Lambda_0\times[-\frac{2(M+\sqrt{g})}{\nu};\frac{2(M+\sqrt{g})}{\nu}]$. 

\begin{Remark}
For any $x\in\mathbb{Z}^d$, there exists $k$ such that $|\hat{V}(x,k)-\lambda|\leq \nu$
\end{Remark}

Hence, there is no way avoiding a resonance of order $\nu$ for all $x$, and we cannot look for good boxes as free of any resonances.
Nevertheless, we  prove that good boxes appears with high probability when $g\ll1$.
Let $p>d$.  
\begin{Prop}[Initialisation of the MSA under the assumption (C1)] \label{InitiationMulti}
Assume that (C1) holds.
For any $\mu>0$, $L^*\in \mathbb{N}$, there exist $\epsilon>0$ and $L\geq L^*$ such that for any $g < \epsilon$, such that if $\nu>\exp(-\frac{1}{g^{\frac{1}{8d+4p}}})$ then 
\begin{equation}
\mathbb{P}(B_L(x) \text{ is a $\mu$-goog box}) > 1 - \frac{1}{L^{2p}}
\end{equation}
where $B_L(x)=x+[-L;L]^d\times [-\frac{M}{\nu};\frac{M}{\nu}]$. 
\end{Prop}

The proof of this proposition will be carried over in Section \ref{SHarmonic}.
For the usual Anderson model, Theorem \ref{Tlocalisation} would follow from (see Theorem 8.3 in  \cite{disertori2008random}): 
\begin{enumerate}
\item MSA initialisation (Theorem 11.1 in \cite{disertori2008random}),
\item Wegner estimate (Theorem 5.23 in \cite{disertori2008random}),
\item Independence of these two properties for two distinct boxes (obvious in the usual model).
\end{enumerate}
As already said, the only peculiarity of our model under assumption (C1) is the special form of the potential. In our case, it will thus be enough to prove 
\begin{enumerate}
\item MSA initialization: Proposition \ref{InitiationMulti},
\item Wegner estimate: Eq.~\eqref{EWegnerFinite} in Proposition \ref{PWegner},
\item Independence : Proposition \ref{PlargeK}.
\end{enumerate}
Theorem \ref{Tlocalisation} is then obtained as Theorem 8.3 in  \cite{disertori2008random}.

\subsection{$L^2$ driving (C2)} 
A new problem appears here: For which distance on $\mathbb{Z}^d\times\mathbb{Z}$ should we prove the exponential decay? 
In the smooth case, $\hat{\Delta}$ was a local operator, so the usual distance on works fine. 
But because $g(k'-k)$ is non-zero for $k-k'$ large if the driving is only in $L^2 ([0,T])$, 
the operator $\hat{\Delta}$ connects now points $(\hat{x},\hat{y})$ that are not close to each other in $\mathbb{Z}^d\times\mathbb{Z}$ and there is no exponential decay along the frequency $k$. 
In order to prove exponential decay on $\mathbb{Z}^d$, we introduce a new decay function on $\mathbb{Z}^d\times\mathbb{Z}$, which can actually easily be used in the ``random walk expansion" that appears in the MSA. 

\begin{Def}
Let $G : \big(\mathbb{Z}^d\times\mathbb{Z} \big)^2\rightarrow \mathbb{R}$ such that for all any $\hat{x}_0\in \mathbb{Z}^d\times\mathbb{Z}$, $G(\hat{x}_0,.)\in L^1(\mathbb{Z}^d\times\mathbb{Z})$ with $\|G(\hat{x}_0,.)\|_{L^1}<1/2$. 
We define the decay function $d_G$ by 
\begin{equation}
d_{G}(\hat{x},\hat{y})=-\ln\Big(\sum_{\mathcal{C}(\hat{x}\rightarrow \hat{y})}\prod_i |G(\hat{z}_i,\hat{z}_{i+1})|\Big)
\end{equation}
for any $\hat{x}$,$\hat{y} \in \mathbb{Z}^d\times\mathbb{Z}$ if $\hat{x}\neq\hat{y}$ and 0 otherwise, 
where $\mathcal{C}(\hat{x}\rightarrow \hat{y})$ is the set of all  finite sequences of the type $(\hat{x} = \hat{z}_0,\hat{z}_1,\hat{z}_2,\dots,\hat{z}_k=\hat{y})$ (or ``paths" from $\hat{x}$ to $\hat{y}$).
\end{Def}

Let $P:\Z^d \times \Z \to \R$ be defined by 
\begin{equation*}
P((x,k))=\begin{cases} 1 / \sqrt{g} \text{ if }k\nu \in [-M-\sqrt{g} , M+\sqrt{g}], \\
\frac{1}{\nu (|k|-1)-M} \text{ if }k\nu \notin [-M-\sqrt{g} , M+\sqrt{g}].
\end{cases}
\end{equation*}
We say that $\hat{x}$ is a resonant site if $| \hat{V}_\omega(\hat{x})-\lambda | < \sqrt{g}$. 
We have defined the function $P(\hat{x})$ such that if there is no resonant site on $x\times \mathbb{Z}$, then $P(\hat{x})>\frac{1}{|\hat{V}_\omega(\hat{x})-\lambda|}$.

\begin{Def}
Under assumption (C2), we say that $C_L(x) = (x+[-L,L]^d)\times\mathbb{Z} \subset\mathbb{Z}^d\times\mathbb{Z}$ is a $\mu$-good column if there exists a decay function $\tilde{d}_G$ such that 
\begin{equation*}
 |(\hat{x},(\hat{H}_{|C_L(x)}-\lambda)^{-1}\hat{y})|\leq P(\hat{x})e^{- \tilde{d}_G(\hat{x},\hat{y})}
\end{equation*} 
for all $\hat{y} \in \partial^{in}C_L(x)$, and such that 
\begin{equation*}
\sum_{\hat{y}\in \partial^{in}C_L(x)}e^{-\tilde{d}_G(\hat{x},\hat{y})}<e^{-\mu L}.
\end{equation*}
\end{Def}


\begin{Prop}[Initialisation of the MSA under the assumption (C2)] \label{InitiMultiGC}
Assume that (C2) holds.
For any $\mu>0$, $L^*\in \mathbb{N}$, there exist $\epsilon>0$ and $L\geq L^*$ such that for any $g < \epsilon$, such that if $\nu>1$ then 
\begin{equation}
\mathbb{P}\big( C_L(x) \text{ is a $\mu$-good column } \big) > 1-\frac{1}{L^{2p}}.
\end{equation}
\end{Prop}

As in the smooth case (C1), Theorem \ref{Tlocalisation} will follow from the Wegner estimate (Eq.~\eqref{EWegnerInfinite} in Proposition \ref{PWegner}) the initialization of the MSA (Proposition \ref{InitiMultiGC}), and the stochastic independence of distinct columns (obvious here). 
But there is still one difference : the MSA has to be performed with infinite columns. 
This issue will be addressed in Section \ref{technical subsection}, where we explain the technicals adaptations to perform in the proof in \cite{disertori2008random}.  

\section{Wegner Estimate} \label{SWegner}
In this Section, we prove Proposition \ref{PWegner} (Wegner estimate). 
For \eqref{EWegnerFinite} (finite column), we closely follow \cite{WegnerEstime}, while for \eqref{EWegnerInfinite} (infinite column), we follow  \cite{ducatez:hal-01271084} (see also \cite{chulaevski}). 
Thanks to the resolvent formula, we have the Shur formula : for any $P$ projector and $B = PBP$, then 
\begin{equation}
P(A+B)^{-1}P = ((PA^{-1}P)^{-1}+B)^{-1}
\end{equation} 
Where the two last ``$\cdot^{-1}$" in the right hand side correspond to the inverse for operators restricted to $Im(P)$.



\begin{proof}[Proof of \eqref{EWegnerFinite}.]
We follow the proof from \cite{disertori2008random}. Let $\Lambda \subset \mathbb{Z}^d\times\mathbb{Z}$, $E\in\mathbb{R}$. 
Let $P_x$, $x\in \mathbb{Z}^d$ the projectors on the subspace $\{ x\}\times [k_0 - K,k_0 + K]$ and $\Lambda_0 \subset \mathbb{Z}^d$ the projection of $\Lambda$ on its first parameters. 
\begin{align*}
  &\mathbb{P}(\exists \bar\lambda \mbox{ eigenvalue of $\hat{H}_{|\Lambda}$ : } \bar\lambda \in [E-\epsilon , E+\epsilon ]) \\  
  &\qquad\leq   \mathbb{E}(\Tr(1_{[E-\epsilon,E+\epsilon]}(\hat{H}_{|\Lambda}))) \\ 
 & \qquad\leq  \mathbb{E}(2\epsilon \Im(\Tr(\hat{H}_{|\Lambda}-E-i\epsilon)^{-1})) \\
 & \qquad =  \mathbb{E}\big[2\epsilon \Im\big(\sum_{x\in\Lambda_0}\Tr \big(P_x(\hat{H}_{|\Lambda}-E-i\epsilon)^{-1} P_x)\big)\big] \\
 & \qquad =   2\epsilon \Im\Big(\sum_{x\in\Lambda_0} \mathbb{E}\big[ \Tr \big(\big( (P_x(\hat{H}_{|\Lambda}-v_x P_x-E-i\epsilon)^{-1} P_x)^{-1}+v_x P_x \big)^{-1}\big) \big] \Big) \\
 & \qquad =  2\epsilon \sum_{x\in\Lambda_0} \mathbb{E}_{V_y : y\neq x}\Big[ \int \Im\Big( \Tr\big( \big( (P_x(\hat{H}_{|\Lambda}-v_x P_x-E-i\epsilon)^{-1} P_x)^{-1}+v_x P_x \big)^{-1}\big) \rho(x) dv_x \Big)\Big] \\
& \qquad =  2\epsilon \sum_{x\in\Lambda_0} \mathbb{E}_{V_y : y\neq x}\Big[ \int \Big( \sum_{\mu_i\in \sigma( (P_x(\hat{H}_{|\Lambda}-v_x P_x-E-i\epsilon)^{-1} P_x)^{-1}) } \Im(\big(\mu_i+v_x\big)^{-1}) \rho(x) dv_x \Big)\Big] \\
 & \qquad\leq  2\epsilon \sum_{x\in\Lambda_0} \mathbb{E}_{V_y : y\neq x} \big(\pi \|\rho\|_\infty (2K+1) \big)  \\
 & \qquad\leq  2\pi \epsilon (2K+1)|\Lambda_0| \|\rho\|_\infty,
\end{align*}
where, to get the last equality, we used that $P_x$ acts as the identity on the subspace generated by $P_x$. 
\end{proof}


\begin{proof}[Proof of \eqref{EWegnerInfinite}]
Let $\Lambda_0$ be a finite subset of $\Z^d$. We make a change of variable for the potential $\alpha=\frac{1}{|\Lambda_0|}\sum_{x\in\Lambda_0}V_{\omega}(x)$. 
As in \cite{ducatez:hal-01271084} (see also \cite{chulaevski}), the conditional probability of $\alpha$ knowing $\tilde{V}(x)=V_{\omega}(x)-\alpha$ for all $x\in \Lambda_0$, 
admits a density $\xi_{\tilde{V}}(\alpha)$ and there exists a constant $C$ such that, on a set $U$ belonging to the sigma-algebra generated by $\tilde{V}(x)$ for all $x\in \Lambda_0$, and with probability larger that $1-C\sqrt{\epsilon}$, 
\begin{equation}
\|\xi_{\tilde{V}}\|_{\infty} \leq \frac{C}{\sqrt{\epsilon}} \big((2M)^{1/2}\|\rho\|_{\infty}+(2M)^{3/2}\|\rho'\|_{\infty}\big)
\end{equation}

Because of the symmetry described in Remark \ref{RSymmetry}, for any realization $(\tilde{V},\alpha_0)$, there exist $\bar\lambda_1,..,\bar\lambda_{|\Lambda_0|}\in [0,\nu]$ such that $\sigma(\hat{H}_{\tilde{V},\alpha_0})=\{\bar\lambda_1,\dots,\bar\lambda_{|\Lambda_0|} \}+\nu \Z$. Now, keeping $\tilde{V}$ fixed and changing $\alpha$, one gets $\sigma(\hat{H}_{\tilde{V},\alpha})=\{\bar\lambda_1+(\alpha-\alpha_0),\dots,\bar\lambda_{|\Lambda_0|}+(\alpha-\alpha_0) \}+\nu \Z$. Then, for any $E\in \R$,  
\begin{align*}
\mathbb{P}(d(\sigma(\hat{H}),E)<\epsilon)& \leq C\sqrt{\epsilon}+\mathbb{P}(\{d(\sigma(\hat{H}),E)<\epsilon \} \cap U) \\
& \leq C\sqrt{\epsilon}+\mathbb{E}_{\tilde{V}}\big(\mathds{1}_U \sum_{i=1}^{\Lambda_0} \sum_{k\in \mathbb{Z}} \int \mathds{1}(|\bar\lambda_i + k\nu+(\alpha-\alpha_0)-E|<\epsilon) \xi_{\tilde{V}}(\alpha) d\alpha \big) \\
& \leq C\sqrt{\epsilon}+2\epsilon \frac{1}{\sqrt{\epsilon}} C\big((2M)^{1/2}\|\rho\|_{\infty}+(2M)^{3/2}\|\rho'\|_{\infty}\big) K_0
\end{align*}
where $K_0$ is the maximum number of eigenvalue $\bar\lambda$ in $\sigma(\hat{H}_{\tilde{V},\alpha_0})$ such that there exists $\alpha\in [-M,M]$ such that  $|\bar\lambda+\alpha-\alpha_0-E|<\epsilon$ with non-zero probability. 
In particular we have $K_0\leq 2|\Lambda_0|\frac{M}{\nu}+1$. 
\end{proof}

\section{Smooth driving (C1)}\label{SHarmonic}
\begin{proof}[Proof of Proposition \ref{InitiationMulti}.]
 Proposition \ref{InitiationMulti} is deduced from Proposition \ref{DecayFromChain} and Proposition \ref{PProbaEst} below.
 \end{proof} 

The key tool for the MSA is the following formula : 
\begin{equation}\label{EHuygens}
(\hat{v}_0, (\hat{H}-\lambda)^{-1} \hat{z}) = \sum_{\hat{u}\in \partial^{in}\Lambda,\hat{v}\in \partial^{ext} \Lambda}  (\hat{v}_0,(\hat{H}_{|\Lambda}-\lambda)^{-1}\hat{u}) (\hat{u},g \hat{\Delta}\hat{v}) (\hat{v},(\hat{H}-\lambda)^{-1}\hat{z})
\end{equation}
for any $\hat{v}_0\in \Lambda$ and $\hat{z} \notin \Lambda$, and $\Lambda \subset \mathbb{Z}^d\times\mathbb{Z}$ with $z\notin \Lambda$, which is a direct application of the well known resolvent formula. We will repeat it as many times as we can, replacing $v$ for $v_0$ and choosing correctly the new $\Lambda$. 
The next subsection deals with this question. 

\subsection{Resonant sites, security box and propagation decay}
Remind that $\hat{v} = (x,k) \in \mathbb{Z}^d\times\mathbb{Z}$  is a resonant site if $|\hat{V_\omega}(\hat{v})-\lambda | = | v_x+\nu k -\lambda | < \sqrt{g}$. 
Obviously, for any $x$ there exits a segment $K_x \subset \mathbb{Z}$ so that $(x,k)$ is a resonant site for $k\in K_x$, where  
$K_x$ is of the form $K_x=\mathbb{Z}\cap [k_0-\sqrt{g}/\nu, k_0+\sqrt{g}/\nu]$ for some $k_0$ that depends on $V_\omega(x)$ (Figure \ref{resonanceSite}). 
Around each segment of resonant sites $K_x$, we define a {security box} $\Lambda_{K_x} = \{z\in \mathbb{Z}^d\times\mathbb{Z} : d(z,K_x)< N\}$,
 where $N$ is an integer that will be defined later, and $d$ is the usual graph distance on $\mathbb{Z}^d\times\mathbb{Z}$.

\begin{figure} \label{resonanceSite}
\begin{center}
\includegraphics[height=7cm,width=10cm]{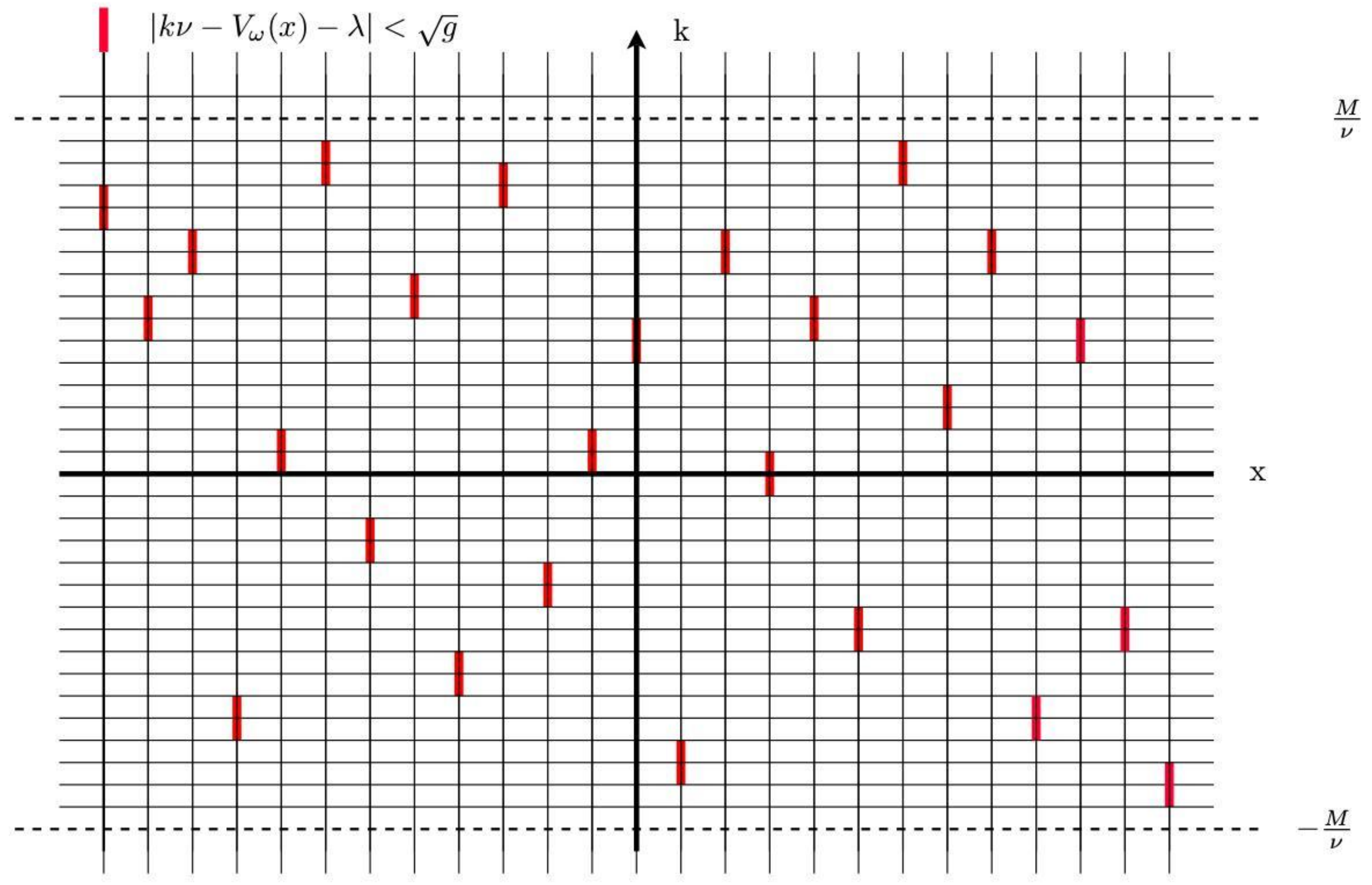}
\caption{resonant sites}
\end{center}
\end{figure}

We will say that a set of the form $\Lambda_0 \times I \subset \Z^d \times \Z$ is not strongly resonant if $d(\sigma(\hat{H}_{|\Lambda_0 \times I}),\lambda) > \nu^2 \alpha(g)$, where $\alpha(g)$ is a function which will be defined at the end of the proof of Proposition \ref{PProbaEst} below. 


\begin{Prop} \label{DecayFromChain}
Let $L\in \mathbb{N}$. 
If no security boxes intersect, if no security box is strongly resonant, and if $(x+[-L,L]^d)\times \Z$ is not strongly resonant, then for any $y\in\partial^{in} (x+[-L,L]^d)$, $k_1$,$k_2 \in \mathbb{Z}$,
\begin{equation} 
 \big((x,k_1),\big(\hat{H}_{|(x+[-L,L]^d)\times[k_0-K,k_0+K]}-\lambda \big)^{-1} (y,k_2)) \leq \frac{2(\sqrt{g}^\frac{N}{2})^{n_0}}{(\nu^2\alpha(g))^2}
 \end{equation} 
 where $n_0= \lfloor\frac{d((x,k_1),(y,k_2))}{2N}\rfloor$.
 \end{Prop}
In particular this proposition implies that $(x+[-L,L]^d)\times[k_0-K,k_0+K]$ is a $\mu$-good box with 
$$\mu = -\Big(\frac{\ln(g)}{4}-\frac{2\ln(\nu^2\alpha(g))}{L}\Big).$$

\begin{proof}
For this proof, we work inside the the space $L^2 ((x+[-L,L]^d)\times[k_0-K,k_0+K])$ and we write simply $\hat{H}$ instead of $\hat{H}_{|(x+[-L,L]^d)\times[k_0-K,k_0+K]}$.


Iterating \eqref{EHuygens}, we obtain the usual random walk expansion for the resolvent (see e.g.\@ \cite{disertori2008random}): 
Given $\hat{x},\hat{y} \in \Z^d\times \Z$, we get 
\begin{multline} \label{chain}
(\hat{x},(\hat{H}-\lambda)^{-1} \hat{y}) =\sum_{\hat{u}_i\in \partial^{in}_i\Lambda(\hat{v}_{i-1}),\hat{v}_i\in \partial^{ext}\Lambda (\hat{v}_{i-1})} \\
(\hat{x},(\hat{H}_{|\Lambda (\hat{v}_0)}-\lambda)^{-1},\hat{u}_1)(\hat{u}_1,g \hat{\Delta} \hat{v}_1)(\hat{v}_1,(\hat{H}_{|\Lambda (\hat{v}_1)}-\lambda)^{-1} \hat{u}_2)(\hat{u}_2,g \hat{\Delta} \hat{v}_2)
\dots (\hat{v}_n,(\hat{H}-\lambda)^{-1}\hat{y} ).
\end{multline}
In this writing, we need to specify when we stop iterating  \eqref{EHuygens} and how $\Lambda (\hat{v}_{i-1})$ is defined. 
The following choice will guarantee the desired exponential decay: 
\begin{enumerate}
\item If $|\hat{v}-\hat{y}|\leq N$, we stop iterating \eqref{EHuygens}. 

\item if $\hat{v}$ is not a resonant site, we choose $\Lambda (\hat{v}) = \{ \hat{v} \}$.
There are then at most $6d+2$ points in $\partial^{ext}\Lambda(\hat{v})$.

\item if $\hat{v}$ is a resonance site, we choose $\Lambda(\hat{v}) = \Lambda_{ K_x}$. 
There are at most $CdN^{d-1}(N+\sqrt{g}/\nu)$ points in $\partial^{ext}\Lambda(\hat{v})$ for some numerical constant $C>0$.
\end{enumerate}
See Figure \ref{SecurityBox} for a typical path from $\hat{x}$ to $\hat{y}$.


From \eqref{chain}, we obtain 
\begin{multline}\label{chain bounded}
| (\hat{x},(\hat{H}-\lambda)^{-1} \hat{y}) |Ê \le \\
\sum 
\big|(\hat{x},(\hat{H}_{|\Lambda (\hat{v}_0)}-\lambda)^{-1},\hat{u}_1)(\hat{u}_1,g \hat{\Delta} \hat{v}_1) \big|
\big|(\hat{v}_1,(\hat{H}_{|\Lambda (\hat{v}_1)}-\lambda)^{-1} \hat{u}_2)(\hat{u}_2,g \hat{\Delta} v_2) \big| \dots 
\|(\hat{H}-\lambda)^{-1}\|.
\end{multline}
The factors in each term in this sum are bounded in two different ways, depending on whether they are resonant or not: 
\begin{enumerate}
\item 
If $\hat{v}_i = (x,k)$ is not a resonant site, then $(\hat{H}_{|\Lambda}-\lambda)= (v_x+k\nu-\lambda) \delta_{(x,k)}$ so that  
\begin{equation} \label{EnonReso}
\big|   \big((x,k),(\hat{H}_{|\Lambda (\hat{v}_i)}-\lambda)^{-1}(x,k)\big)    \big((x,k),g \hat{\Delta}(x',k')\big)  \big|  \le \frac{\big|\big((x,k),g \hat{\Delta}(x',k')\big)\big|}{\sqrt{g}} \leq \sqrt{g}.
\end{equation} 

\item 
If $\hat{v}_i = (x,k)$ belongs to $K_x$, then 
\begin{equation} \label{EReso}
\big|   \big((x,k),(\hat{H}_{|\Lambda(\hat{v}_i)}-\lambda)^{-1}(x',k')\big)    \big((x',k'),g \hat{\Delta}(x'',k'')\big)  \big|  \le  \frac{g}{d(\sigma(\hat{H}_{|\Lambda_{K_x}}),\lambda)} .
\end{equation} 
\end{enumerate}

The sum in \eqref{chain bounded} will be small, if for every path joining $\hat{x}$ to $\hat{y}$, the number $n$ of non resonant sites is large enough to dominate the resonant terms (indexed by $J$), i.e.\@
\begin{equation}
(2(d+1)\sqrt{g})^n \ll \prod_{j\in J} d(\sigma(\hat{H}_{|\Lambda_j}),\lambda)
\end{equation} 
We can now understand the reason why we have introduced the security boxes: 
Assuming that no security boxes intersect one to another, then $u_i$ is a resonant site implies that $u_{i+1}$ is not resonant and its distance to any resonant sites is at least larger than $N$. 
From this we can deduce that for any path joining $\hat{x}$ to $\hat{y}$, every resonant term is followed by at least $N$ non resonant ones. Let $N\in \mathbb{N}$ such that 
\begin{equation} \label{Ndeter}
\frac{N^{d-1}((2N+\frac{\sqrt{g}}{\nu}))(2d+2)^{N+1}( \sqrt{g})^{\frac{N-1}{2}}}{\nu^2 \alpha(g)} <  1  
\end{equation}  
Then, if $\hat{u}_i$ is resonant, and assuming that, there is no strongly resonant security box, and no intersecting security boxes, we find that the following product of $N+1$ consecutive factors can be bounded as 
\begin{multline*}
\big|(\hat{v}_i,(\hat{H}_{|\Lambda (\hat{v}_i)}-\lambda)^{-1} \hat{u}_{i+1})(\hat{u}_{i+1},g \hat{\Delta} v_{i+1}) \big| 
\dots 
\big|(\hat{v}_{i+N},(\hat{H}_{|\Lambda (\hat{v}_{i+N})}-\lambda)^{-1} \hat{u}_{i+N+1})(\hat{u}_{i+N+1},g \hat{\Delta} v_{i+N+1}) \big| \\
\le  \frac{(\sqrt{g})^N}{d(\sigma(\hat{H}_{|\Lambda_{K_x}}),\lambda)} \leq \frac{(\sqrt{g})^{\frac{N+1}{2}}}{N^{d-1}((2N+\frac{\sqrt{g}}{\nu}))(2(d+1))^N} .
\end{multline*}
Hence, for a path connecting $\hat{x}$ to $\hat{y}$ in $l=k(N+1)+s$ steps ($s<N+1$), we obtain 
\begin{multline*}
\big|(\hat{x},(\hat{H}_{|\Lambda (\hat{v}_i)}-\lambda)^{-1} \hat{u}_{1})(\hat{u}_{1},g \hat{\Delta} v_{1}) \big| 
\dots 
\big|(\hat{v}_{i+l-1},(\hat{H}_{|\Lambda (\hat{v}_{l-1})}-\lambda)^{-1} \hat{u}_{l})(\hat{u}_{l},g \hat{\Delta} v_{l}) \big| \\
\leq \Big(\frac{(\sqrt{g})^{\frac{N+1}{2}}}{N^{d-1}((2N+\frac{\sqrt{g}}{\nu}))(2(d+1))^N}\Big)^k\frac{(\sqrt{g})^{s-1}}{\nu^2 \alpha(g)} .
\end{multline*}

We can now conclude the proof.
Indeed, any path connecting $\hat{x}$ to $\hat{y}$ contains at least $(d((x,k_1),(y,k_2))-N)/2$ steps. 
Denoting by $A_l$ the set of paths connecting $\hat{x}$ to $\hat{y}$ in $l$ steps, we find 
\begin{align*}
| (\hat{x},(\hat{H}-\lambda)^{-1} \hat{y}) | & \leq \sum_{l=(d((x,k_1),(y,k_2))-N)/2}^\infty |A_l|   \Big(\frac{(\sqrt{g})^{\frac{N+1}{2}}}{N^{d-1}((2N+\frac{\sqrt{g}}{\nu}))(2(d+1))^N}\Big)^k\frac{(\sqrt{g})^{s-1}}{\nu^2 \alpha(g)} \frac{1}{{\nu^2 \alpha(g)}}\\
& \leq \sum_{l=(d((x,k_1),(y,k_2))-N)/2}^\infty \sqrt{g}^{l/2} \frac{1}{(\nu^2 \alpha(g))^2} \\
& \leq \frac{(\sqrt{g}^\frac{N}{2})^{n_0}}{(1-\sqrt{g}) (\nu^2\alpha(g))^2}
\end{align*}
\end{proof}

\begin{figure} \label{SecurityBox}
\begin{center}
\includegraphics[height=7cm,width=10cm]{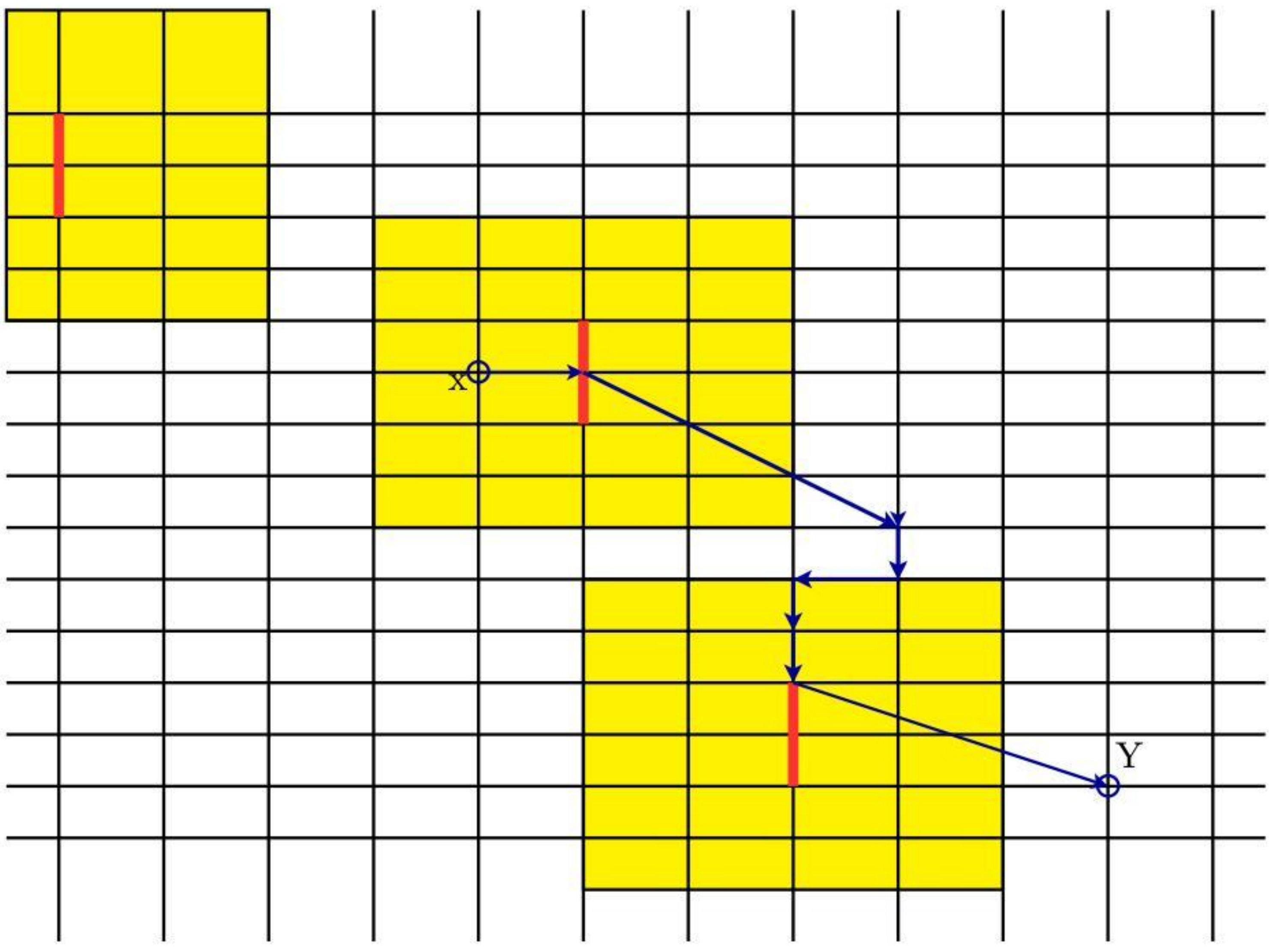}
\caption{A typical path from $\hat{x}$ to $\hat{y}$. In red the resonant sites and in yellow the security boxes with $N=2$.}
\end{center}
\end{figure}


\begin{proof}[Proof of Proposition \ref{PlargeK}]
For any $\hat{x}\in \Lambda_0\times [k_0-K;k_0+K]$, $|\hat{V}(\hat{x})-\lambda|\geq \sqrt{g}$. One can now do the random walk development as previously with no resonant term.
\end{proof}


\begin{Prop}\label{PProbaEst}
The probability of the event ``there is no strongly resonant security box, and no intersecting security boxes" is smaller than $1/L^{2d}$ when $g$ goes to 0 assuming $N =\mathcal O(\frac{\ln(\nu)}{\ln(g)})$, $L=m_1 N$, with $m_1$ a fixed large integer and $|\ln(\nu)|\leq g^{-\frac{1}{8d+4p}}$. 
\end{Prop}
\begin{proof}
To deal with the strongly resonant boxes, we use the Wegner type estimate \eqref{EWegnerFinite} with $\epsilon = \nu^2 \alpha(g)$:
\begin{align}
& \mathbb{P}(\Lambda_{K_x} \text{ is strongly resonant }) \nonumber \\
& \qquad \leq \sum_{k_0 = -M/\nu}^{M/\nu} \mathbb{P}(\Lambda_{K_x} \text{ is strongly resonant and }  K_x=\mathbb{Z}\cap [k_0-1/(\nu \sqrt{g} ), k_0+1/(\nu \sqrt{g})]) \nonumber \\
 & \qquad \leq \sum_{k_0 = -M/\nu}^{M/\nu} \mathbb{P}(\Lambda_{\mathbb{Z}\cap [k_0-1/(\sqrt{g}\nu) , k_0+1/(\sqrt{g}\nu)]} \text{ is strongly resonant}) \nonumber\\
 &  \qquad \leq \frac{2M}{\nu} 2\pi \nu^2\alpha(g) (N^d (\frac{2\sqrt{g}}{\nu}+2N)) \|\rho\|_\infty \nonumber \\
 & \qquad \leq 4 M (N^d (\frac{2\sqrt{g}}{\nu}+2N)) \nu \alpha(g) 
\end{align}

We deal now with the probability of non intersecting security boxes:  For any $x,y \in [-L,L]^d$, $\Lambda_{K_x} \cap \Lambda_{K_y} = \emptyset $. This will be satisfied if there is no $ |k|\leq 2N$ such that $|v_x-v_y+k\nu|\leq \sqrt{g}$. If $\nu\leq \sqrt{g}$, the probability $P$ of intersecting security boxes is bounded by: 
\begin{align}
 \label{EnInterac1}
P &\leq\frac{(2L)^d\big((2L)^d-1\big)}{2}\mathbb{P}\big(|v_x-v_y|< 2 (N \nu + \sqrt{g} ) \big) \nonumber \\ & \leq 2 (2L)^d\big((2L)^d-1\big) (N \nu + \sqrt{g} ) \|\rho\|_\infty 
\end{align}
and in any case (when $\nu>\sqrt{g}$) by 
\begin{equation} \label{EnInterac2}
P \leq 2 (2L)^d\big((2L)^d-1\big) (N + 1)\sqrt{g}  \|\rho\|_\infty 
\end{equation}

From \eqref{EnInterac1} (or \eqref{EnInterac2}) and Proposition \ref{DecayFromChain} we conclude the proof of our theorem. We need:
\begin{equation}\label{ThreeCondition}
\begin{cases}
4 M (N^d (\frac{2\sqrt{g}}{\nu}+2N)) \nu \alpha(g) \leq \frac{1}{2L^{2p}} \\
2 (2L)^d\big((2L)^d-1\big) (N \nu + \sqrt{g} ) \|\rho\|_\infty \leq \frac{1}{2L^{2p}}
\\ -(\frac{\ln(g)}{4}-\frac{2\ln(\nu^2\alpha(g))}{L})>\mu
\end{cases}
\end{equation}
or (when $\nu>\sqrt{g}$)
\begin{equation}
\begin{cases}
4 M (N^d (\frac{2\sqrt{g}}{\nu}+2N)) \nu \alpha(g) \leq \frac{1}{2L^{2p}} \\
2 (2L)^d\big((2L)^d-1\big) (N+1) \sqrt{g}  ||\rho||_\infty \leq \frac{1}{2L^{2p}} \\
-(\frac{\ln(g)}{4}-\frac{2\ln(\nu^2\alpha(g))}{L})>\mu
\end{cases}
\end{equation}
and \eqref{Ndeter}. We set $\alpha(g)=1$ in case of $\nu<\sqrt{g}$ and $\alpha(g)=g$ in case of $\nu>\sqrt{g}$.

\begin{enumerate}
\item $N= n_1\frac{\ln(\nu)}{\sqrt{g}}$ with $n_1>7$. 
\item $L=m_1 N$ with $m_1$ a large enough integer.
\end{enumerate}
We have then  $-(\frac{ln(g)}{4}-\frac{ln(\nu^2\alpha(g))}{L})>|\ln(g)|(\frac{1}{4}-\frac{1}{m_1})$. Then assume $|\ln(\nu)|\leq g^{-\frac{1}{8d+4p}}$. So we get $L^{4d+2p}\sqrt{g} = O(g^{1/4})$. 
Finally the three conditions of \eqref{ThreeCondition} are satisfied in the limit $g\rightarrow 0$ and this is the end of the proof of \ref{InitiationMulti}.
\end{proof}

\section{$L^2$ driving (C2)}\label{section: general L2 case}

We now consider the case of an $L^2$ driving. 
In this set up, we will work on infinite columns $C_L(x)=(x+[-L,L]^d)\times\mathbb{Z}$, so that distinct column are independent with respect to the disorder. 
Instead, one should be careful in the random walk expansion since infinite sums appear. 
That this is not a problem comes from the decay of the Green function at the large frequencies:

\subsection{Decay of the Green function along the frequency axes}\label{SS Decay Green}
\begin{Prop}
Let $\hat{\phi}$ be an eigenfunction of $\hat{H}$ with eigenvalue $\bar\lambda$. Then 
\begin{equation}\label{eigDec1}
\sum_{x,k}||k\nu-\lambda|\hat{\phi}(x,k)|^2   \le  (g+M)^2.
\end{equation}
In particular
\begin{equation}\label{eigDec2}
|\hat{\phi}(x,k)| \leq \frac{1+M+g}{1+|k\nu-\bar\lambda|}
\end{equation}
for any $x$.
\end{Prop}
\begin{proof}
We use the time representation of $\hat{\phi}$. 
Recall that $\phi(t)=e^{i\lambda t}\psi(t)$ with $\psi$ solution of \eqref{Ealpha}. 
Since the evolution is unitary, for all $t \in [0,T], \|\phi(t)\|=\|\psi(t)\|=\|\psi(0)\|=\|\phi(0)\|$.  So  
\begin{align*}
\sum_{x,k}||k\nu-\lambda|\hat{\phi}(x,k)|^2 & 
=
\frac{1}{T} \int_0^T \|(i\partial_t-\bar\lambda)\phi(t)\|^2 dt \\ & =\frac{1}{T}\int_0^T \|(-g\Delta(t) + V)\phi(t)\|^2 dt \\ & \leq \frac{1}{T}\int_0^T \|(g\Delta(t) + V)\|^2 dt \\
& \leq g^2+2M \frac{1}{T}\int_0^T \|(g\Delta(t)\| dt+M^2 \\
& \leq (g+M)^2,
\end{align*}
and we deduce that $(1+(|k\nu-\lambda|))\hat{\phi}(x,k)$ is square integrable. 
\end{proof}

From this we can deduce an estimate for the resolvent :

\begin{Prop}\label{Fcontrol}
There exist a constant $C$ depending only on $\nu$ so that we have
\begin{align*}
 |(\hat{z},(\hat{H}_{|C_L(x)}-\lambda)^{-1}\hat{y})| \leq \frac{(2L+1)^{d/2} (2+M)P(\hat{x})}{1+|k_z-k_y|}\Big(\sup_i \frac{1}{|\lambda-\bar\lambda_i|}+C\Big)
\end{align*}
for any $\hat{z} = (z,k_z)$, $\hat{y} = (y,k_y) \in C_L(x)$, where $\bar\lambda_i$ are the eigenvalue of $\hat{H}_{|C_L(x)}$.
\end{Prop}

\begin{proof} We decompose $\hat{H}_{|C_L(x)}$ into its eigenvectors and we apply Cauchy Schwartz. 
The eigenvalues of $\hat{H}_{|C_L(x)}$ are all of the form $\bar\lambda_i+k\nu$, where we can assume that $\bar\lambda_i$ are such that $|\bar\lambda_i+k\nu-\lambda|\geq \nu/2$ if $k\neq 0$. Then
\begin{align*}
& (\hat{z},(\hat{H}_{|C_L(x)}-\lambda)^{-1}\hat{y}) \\ 
& \qquad = \sum_{i=1}^{|\Lambda|}\sum_{k\in\mathbb{Z}} \frac{1}{\bar\lambda_i+\nu k-\lambda} \phi_{\bar\lambda_i+\nu k}(\hat{z})\phi_{\bar\lambda_i+\nu k}(\hat{y}) \\
& \qquad \leq \Big( \sum_{i=1}^{|\Lambda|}\sum_{k\in\mathbb{Z}} (1+|\bar\lambda_i+\nu(k-k_z)|)^2 |\phi_{\bar\lambda_i+\nu k}(\hat{z})|^2 \Big)^{1/2}. \\ 
& \qquad\Big( \sum_{i=1}^{|\Lambda|}\sum_{k\in\mathbb{Z}} \frac{1}{|\bar\lambda_i+\nu k-\lambda|^2} \frac{1}{(1+|\bar\lambda_i+\nu(k-k_z)|)^2} |\phi_{\bar\lambda_i+\nu k}(\hat{y})|^2  \Big)^{1/2} \\
& \qquad = \Big( \sum_{i=1}^{|\Lambda|}\sum_{k\in\mathbb{Z}} (1+|\bar\lambda_i+\nu(k-k_z)|)^2 |       \phi_{\bar\lambda_i}(z,k_z-k)      |^2 \Big)^{1/2}. \\ 
& \qquad\Big( \sum_{i=1}^{|\Lambda|}\sum_{k\in\mathbb{Z}} \frac{1}{|\bar\lambda_i+\nu k-\lambda|^2} \frac{1}{(1+|\bar\lambda_i+\nu(k-k_z)|)^2} |\phi_{\bar\lambda_i+\nu k}(\hat{y})|^2  \Big)^{1/2} \\
%
\end{align*}
We use now \eqref{eigDec1} to control the first factor, and \eqref{eigDec2} to get an estimate on  $ |\phi_{\bar\lambda_i+\nu k}(\hat{y})|$ in the second one: 
\begin{align*}
& (\hat{z},(\hat{H}_{|C_L(x)}-\lambda)^{-1}\hat{y}) \\ 
&  \qquad \leq (1+M+g)^2\big(  \sum_{i=1}^{|\Lambda|}\sum_{k\in\mathbb{Z}} \frac{1}{|\bar\lambda_i+\nu k-\lambda|^2} \frac{1}{(1+|\bar\lambda_i+\nu(k-k_z)|)^2} \frac{1}{(1+|\bar\lambda_i+\nu (k-k_y)|)^2} \big)^{1/2} \\
& \qquad = (1+M+g)^2\big( \sum_{i=1}^{|\Lambda|}\frac{1}{|\bar\lambda_i-\lambda|^2} \frac{1}{(1+|\bar\lambda_i+\nu(k-k_z)|)^2}\frac{1}{(1+|\bar\lambda_i+\nu k_y|)^2} \\ & \qquad + \sum_{i=1}^{|\Lambda|}\sum_{k\in\mathbb{Z}^*} \frac{1}{|\bar\lambda_i+\nu k-\lambda|^2}\frac{1}{(1+|\bar\lambda_i+\nu(k-k_z)|)^2} \frac{1}{(1+|\bar\lambda_i+\nu (k-k_y)|)^2} \big)^{1/2} \\
& \qquad \leq |\Lambda|^{1/2}(1+M+g)^2(\sup_i \frac{1}{|\lambda-\bar\lambda_i|}+C)P(\hat{z})\frac{1}{(1+|k_z-k_y|)},
\end{align*}
where the last inequality comes from the estimate of the integral 
$$\int dk \frac{1}{1+k^2}\frac{1}{1+(k-k_z)^2}\frac{1}{1+(k-k_y)^2} \sim \frac{1}{(1 + |k_z|)^2}\frac{1}{(1 + |k_z - k_y|)^2}.$$
\end{proof}

\begin{Def}
We say that $C_L(x)$ is not strongly resonent if \begin{equation}
\inf_{\bar\lambda_i\in \sigma(\hat{H}_{|C_{L}(x)})}\{|\bar\lambda_i-\lambda|\}>e^{-\sqrt{L}}.
\end{equation}
\end{Def}
In particular, if $C_L(x)$ is not strongly resonant, we have
\begin{align*}
 |(\hat{z},(\hat{H}_{|C_L(x)}-\lambda)^{-1}\hat{y})| \leq \frac{C L^{d/2} P(\hat{z})}{1+|k_z-k_y|}e^{\sqrt{L}}
\end{align*}
where $C$ is a constant.

\subsection{The decay function}
If Anderson localization is most of the time studied over $\mathbb{Z}^d$, the problem could be raised on any set of point $X$. It is indeed easy to define a random potential $V(x)$, $x\in X$ and a ``Laplacian" $\Delta(x_1,x_2)$ without assuming a particular geometry of the system. But to recover the decay, one should then first define a decay function, and $\Delta$ is the only object that we can use to construct such a decay function. We first give a general definition.
\begin{Def} \label{Ddistance} Let $G:X\times X\rightarrow \mathbb{R}_+$, for any $\hat{x}$, $\hat{y} \in X $,
\begin{equation}
d_{G}(\hat{x},\hat{y})=-\ln\Big(\sum_{\mathcal{C}(\hat{x}\rightarrow \hat{y})}\prod_i |G(\hat{z}_i ,\hat{z}_{i+1})|\Big)
\end{equation}
if $\hat{x}\neq\hat{y}$ and 0 otherwise, where $\mathcal{C}(\hat{x}\rightarrow \hat{y})$ is the set of all paths $\hat{x}=\hat{z}_0,\hat{z}_1,\hat{z}_2,...,\hat{z}_k=\hat{y}$ from $\hat{x}$ to $\hat{y}$.  
\end{Def}


\begin{Prop}
If for any $z\in X$, $\sum_{z_1}|G(z,z_1)|<1/2$, then $d_{G}$ is positive and satisfies the triangle inequality.
\end{Prop}

\begin{proof}

We first check that $d_G$ is positive. Let $\hat{x}$, $\hat{y}$
\begin{align*}
& \sum_{\mathcal{C}(\hat{x}\rightarrow \hat{y})}\prod_i |G(\hat{z}_i ,\hat{z}_{i+1})|  \\ & \qquad \leq  \sum_{\hat{y}'}\sum_{\mathcal{C}(\hat{x}\rightarrow \hat{y}')} \prod_i |G(\hat{z}_i ,\hat{z}_{i+1})| \\
& \qquad  \leq  \sum_{n>0} \prod_{i=0}^n\big(\max_{\hat{z}_i}\sum_{\hat{z}_{i+1}\in X }|G(\hat{z}_i ,\hat{z}_{i+1})| \big) \\
& \qquad  =  \sum_{n>0} \big(\max_{\hat{x}}\sum_{\hat{y}\in X} |G(\hat{x},\hat{y})| \big)^n \\
& \qquad =  \frac{\big(\max_{\hat{x}}\sum_{\hat{y}\in X} |G(\hat{x},\hat{y})| \big)}{1-\big(\max_{\hat{x}}\sum_{\hat{y}\in X} |G(\hat{x},\hat{y})| \big)} \\
& \qquad <  1.
\end{align*}
We now check the triangle inequality. Let $\hat{z}$ be another point in $X$.
\begin{align*}
& d_{G}(\hat{x},\hat{y})+d_{G}(\hat{y},\hat{z}) \\ & \qquad = -\ln\Big(\sum_{\mathcal{C}(\hat{x}\rightarrow \hat{y})}\prod_i |G(\hat{z}_i ,\hat{z}_{i+1})|\Big)-\ln\Big(\sum_{\mathcal{C}(\hat{y}\rightarrow \hat{z})}\prod_j |G(\hat{z}_j, \hat{z}_{j+1})|\Big) \\
& \qquad = -\ln\Big(\sum_{\mathcal{C}(\hat{x}\rightarrow \hat{y})}\sum_{\mathcal{C}(\hat{y}\rightarrow \hat{z})}\prod_i |G(\hat{z}_i, \hat{z}_{i+1})|\prod_j |G(\hat{z}_j, \hat{z}_{j+1})| \Big) \\
& \qquad \geq  -\ln\Big(\sum_{\mathcal{C}(\hat{x}\rightarrow \hat{z})}\prod_i |G(\hat{z}_i ,\hat{z}_{i+1})| \Big) \\
& \qquad  =  d_{G}(\hat{x},\hat{z}).
\end{align*}

\end{proof}

\subsection{initialisation of the multiscale}
\begin{proof}[Proof of Proposition \ref{InitiMultiGC}.]
Proposition \ref{InitiMultiGC} follows from Propositions \ref{MSAInit} and \ref{GCaseInitPro} below. 
\end{proof}

\begin{Def}
We will use $d_G$ with $X=\mathbb{Z}^d\times\mathbb{Z}$ and 
\begin{equation}
G(\hat{x} , \hat{y})=g |\hat{\Delta}(\hat{x},\hat{y})P(\hat{y}) |
\end{equation}
\end{Def}

Remark that we also have $\Delta(\hat{z},.)P(.)\in L^1$ because $\Delta(\hat{z},.)\in L^2$ and $P(.)\in L^2$. 
We will write $\|G\|_{\ell^1 max}=\sup_x \sum_{y}G(x,y)$. This quantity goes to zero as $g\rightarrow 0$. 
The decay function is related to usual distance on $\mathbb{Z}^d$ through the following proposition:  

\begin{Prop} \label{PSommeDist}
For any $\hat{x}=(x,k_x)$,
\begin{equation}
\sum_{z : |x-z|=L}\sum_k e^{-d_G((x,k_x),(z,k))} \leq e^{L\ln((\|G\|_{\ell^1 max})-\ln(1-\|G\|_{\ell^1 max})}
\end{equation}
in particular $\hat{z}=(z,k_z)$, $|x-z|>L$.
\begin{equation}
d_{G}(\hat{x},\hat{z})\geq L(-\ln(\|G\|_{\ell^1 max}))+\ln(1-\|G\|_{\ell^1 max})
\end{equation}
\end{Prop}
\begin{proof}
Because no path of length smaller than $L$ connect $\hat{x}$ with the boundary of $\{(z,k):|x-z|>L\}$,
\begin{equation}\label{Esommabilite}
\sum_{\mathcal{C}(\hat{x}\rightarrow \hat{z})}\prod_i |G(\hat{z}_i \hat{z}_{i+1})| \leq \sum_{n>L}\|G\|_{\ell^1 max}^n \leq \frac{\|G\|_{\ell^1 max}^L}{1-\|G\|_{\ell^1 max}}.
\end{equation}
So
\begin{equation*}
d_{G}(\hat{x},\hat{z}) \geq -L\ln((\|G\|_{\ell^1 max})+\ln(1-\|G\|_{\ell^1 max}).
\end{equation*}
\end{proof}

\begin{Prop} \label{MSAInit}
If there is no resonant site at all in $C_L(x)$, and if $\hat{H}_{|C_L(x)}$ has no eigenvalue $\bar\lambda_i$ with $|\bar\lambda_i-\lambda|\leq \sqrt{g}$, then $C_L(x) \text{ is a $(\mu', \tilde{d}_{G})$ good column }$
\end{Prop}

\begin{proof} 
We use here again the resolvent formula: 
\begin{align*}
(\hat{x},(\hat{H}_{|C_L(x)}-\lambda)^{-1}\hat{y})=\sum_{\hat{z}}\frac{g\hat{\Delta}(\hat{x},\hat{z})}{\hat{V}(\hat{x)}-\lambda} (\hat{z},(\hat{H}_{|C_L(x)}-\lambda)^{-1}\hat{y}).
\end{align*}
Applying it several times yields the usual random walk expansion: 
\begin{align*}
&(\hat{x},(\hat{H}_{|C_L(x)}-\lambda)^{-1}\hat{y})\\ & \qquad =\sum_{\hat{z},\hat{z}_1,\hat{z}_2,...,\hat{z}_n} \frac{g\hat{\Delta}(\hat{x},\hat{z}_1)}{\hat{V}(\hat{x})-\lambda}\frac{g\hat{\Delta}(\hat{z}_1,\hat{z}_2)}{\hat{V}(\hat{z}_1)-\lambda}
\dots     \frac{g\hat{\Delta}(\hat{z}_{n-1},\hat{z}_n)}{\hat{V}(\hat{z}_{n-1})-\lambda}       (\hat{z}_n,(\hat{H}_{|C_L(x)}-\lambda)^{-1}\hat{y})
\end{align*}
Because there is no resonant site, $\frac{1}{\hat{V}(\hat{z})-\lambda}\leq P(\hat{z})$ for any $\hat{z}\in C_L(x)$. So
\begin{align*}
&|(\hat{x},(\hat{H}_{|C_L(x)}-\lambda)^{-1}\hat{y}) |\\ 
& \qquad =P(\hat{x})\sum_{\hat{z},\hat{z}_1,\hat{z}_2,\dots,\hat{z}_n} | g\hat{\Delta}(\hat{x},\hat{z_1})P(\hat{z}_1)g\hat{\Delta}(\hat{z}_1,\hat{z}_2) 
\dots P(\hat{z}_{n-1}) g\hat{\Delta}(\hat{z}_n,\hat{z}_{n-1})(\hat{z_n},(\hat{H}_{|C_L(x)}-\lambda)^{-1}\hat{y}) | \\
&\qquad \leq C P(\hat{x})\sum_{\hat{z},\hat{z}_1,\hat{z}_2,\dots,\hat{z}_n} |g\hat{\Delta}(\hat{x},\hat{z_1})P(\hat{z}_1)g\hat{\Delta}(\hat{z}_1,\hat{z}_2)\dots P(\hat{z}_{n-1}) g\hat{\Delta}(\hat{z}_n,\hat{z}_{n-1}) P(\hat{z}_n)| \frac{L^{d/2}}{\sqrt{g}}\\
& \qquad \leq C L^{d/2} \frac{P(\hat{x})}{\sqrt{g}} \sum_{\mathcal{C}(x\rightarrow y)} \prod_i g |\hat{\Delta}(\hat{z}_i,\hat{z}_{i+1})|P(\hat{z}_{i+1})
\end{align*}
where the first inequality is obtained through Proposition \ref{Fcontrol}  and the hypothesis on the eigenvalues $\bar\lambda_i$. 
So one has
\begin{equation*}
| (\hat{x},(\hat{H}_{|C_L(x)}-\lambda)^{-1}\hat{y}) | \leq C L^{d/2} \frac{P(\hat{x})}{\sqrt{g}} e^{- \tilde{d}_G(\hat{x},\hat{y})}
\end{equation*}
\end{proof}

\begin{Prop}\label{GCaseInitPro}
The probability of the event ``there is no resonant site at all in $C_L(x)$, and $\hat{H}_{|C_L(x)}$ has no eigenvalue $\lambda_i$ with $|\lambda_i-\lambda|\leq \sqrt{g}$" goes to 0 with $g\rightarrow0$ .
\end{Prop}

\begin{proof}
First, 
\begin{equation}
\mathbb{P}(\text{there is no resonant site in $C_L(x)$})\leq ||\rho||_\infty\frac{2M}{\nu}(2L+1)^d \sqrt{2g}. 
\end{equation}
Next, thanks to Wegner estimate,
\begin{equation}
\mathbb{P}(\text{$C_L(x)$ is not strongly resonant })\leq  ||\rho||_\infty\frac{2M}{\nu}(2L+1)^d  \sqrt{2g}.
\end{equation}
This gives the proposition for $g\to 0$.
\end{proof}

\subsection{Technical results for the iteration of the MSA}\label{technical subsection}
We have proved that for a fixed $L$, $ C_L(x)$ is a good column with high probability. 
MSA induces that the property is valid for all $L_{k}$ with $L_{k+1}=L_k^\alpha$, $L_0=L$, but some adaptations with wrt.\@ \cite{disertori2008random} are needed, due to the long range hopping along the frequency axis. 
It turns out that only Theorems 10.14 and 10.20 need to be re-investigated. 
Here we prove Proposition \ref{MSAColumn} below that will play the role of Theorem 10.14 in \cite{disertori2008random} (the equivalent of Theorem 10.20 in \cite{disertori2008random} can then be obtained without any new idea). 

Thanks to the estimates on Green function obtained in Section \ref{SS Decay Green}, we obtain 
\begin{Prop}\label{Pi MSA technical}
\begin{equation}
\sup_{\hat{x},y}\sum_{k_y}\sum_{\hat{z}}\frac{1}{1+|k_x-k_y|}|\Delta(\hat{y},\hat{z})P(\hat{z})|<\infty
\end{equation}
In particular $G(\hat{x},.)=\sum_{k_y}\frac{1}{1+|k_x-k_y|}|\Delta(\hat{y},.)P(.)|$ is in $L^1$ uniformly in $x$.
\end{Prop}
\begin{proof}
We have $\sqrt{|\Delta(\hat{y},.)|} \in L^4$, with a norm that can be bounded uniformly in $\hat{y}$, $\frac{1}{1+|.|}\in L^\frac{4}{3}$ and $P(.)\in L^\frac{4}{3}$. 
\begin{align*}
&\sup_{\hat{x},y}\sum_{k_y}\sum_{\hat{z}}\frac{1}{1+|k_x-k_y|}|\Delta(\hat{y},\hat{z})P(\hat{z})|
\\ & \qquad \leq \Big( \sup_{\hat{x},y,\hat{z}}\sum_{k_y}\frac{1}{1+|k_x-k_y|}\sqrt{|\Delta(\hat{y},\hat{z})|} \Big)\Big(\sup_{\hat{y}} \sum_{\hat{z}} \sqrt{|\Delta(\hat{y},\hat{z})|} P(z) \Big)
\\ & \qquad \leq \big( \| \frac{1}{1+|.|}\|_{L^{\frac{4}{3}}} \|\sqrt{|\Delta(\hat{y},.)|}\|_{L^4} \big)\big(\|\sqrt{|\Delta(\hat{y},.)|}\|_{L^4} \|P(.)\|_{L^{\frac{4}{3}}} \big)
\\ & \qquad < \infty
\end{align*}
\end{proof}



\begin{Prop} \label{MSAColumn}
If there is no two distinct small scale columns $C_{L_k}(y) \subset C_{L_{k+1}}(x)$ which are not $\mu$-good, and there is no columns $C_{2L_k}(y') \subset C_{L_{k+1}}(x)$ that are strongly resonant and $C_{L_{k+1}}(x)$ is not strongly resonant, then  
$C_{L_{k+1}}(x)$ is $\mu'$ good with $\mu'>\mu-\frac{3L_k}{L_{k+1}}$.
\end{Prop}
\begin{proof}
Let $d_G$ the decay function used for the small scale good boxes.
In the case of $C_{L_k}$ is a bad column, we use the resolvent development twice
\begin{align*}
&|(\hat{x},(\hat{H}_{|C_{L_{k+1}}(x)}-\lambda)^{-1}\hat{y})| \\ & \qquad \leq \sum_{\substack{ \hat{z}_1\in \partial^{in} C_{L_{2k}}(x),\\ \hat{z}_2\in \partial^{ext} C_{L_{2k}}(x)}} |(\hat{x},(\hat{H}_{|C_{2L_k}(x)}-\lambda)^{-1}\hat{z}_1)g\hat{\Delta}(\hat{z_1},\hat{z_2})(\hat{z_2},(\hat{H}_{|C_{L_{k+1}}(x)}-\lambda)^{-1}\hat{y})|
\\ & \qquad \leq \sum_{\substack{\hat{z}_1\in \partial^{in} C_{L_{2k}}(x) \\ \hat{z}_2\in \partial^{ext} C_{L_{2k}}(x)}}\sum_{\substack{\hat{z}_3\in \partial^{in} C_{L_{k}}(z_2) \\ \hat{z}_4\in \partial^{ext} C_{L_{k}}(z_2)}}|(\hat{x},(\hat{H}_{|C_{2L_k}(x)}-\lambda)^{-1}\hat{z}_1)g\hat{\Delta}(\hat{z_1},\hat{z_2}) 
\\ & (\hat{z_2},(\hat{H}_{|C_{L_k}(x)}-\lambda)^{-1}\hat{z}_3)g\hat{\Delta}(\hat{z_3},\hat{z_4})(\hat{z_4},(\hat{H}_{|C_{L_{k+1}}(x)}-\lambda)^{-1}\hat{y})|
\\ & \qquad \leq P(\hat{x})  \sum_{\substack{\hat{z}_1\in \partial^{in} C_{L_{2k}}(x) \\ \hat{z}_2\in \partial^{ext} C_{L_{2k}}(x)}}\sum_{\substack{\hat{z}_3\in \partial^{in} C_{L_{k}}(z_2) \\ \hat{z}_4\in \partial^{ext} C_{L_{k}}(z_2)}} e^{\sqrt{L_k}} \frac{C (2L_k)^{d/2}}{1+|k_{\hat{x}}-k_{\hat{z}_1}|}|g\hat{\Delta}(\hat{z_1},\hat{z_2})|\\ 
& P(\hat{z}_2)e^{-d_G(\hat{z_2},\hat{z_3})} g|\hat{\Delta}(\hat{z_3},\hat{z_4})(\hat{z_4},(\hat{H}_{|C_{L_k}(x)}-\lambda)^{-1}\hat{y})|
\end{align*}

So let us define $G'$ as follows:
$$G'(\hat{x},\hat{y})= e^{-d_G(\hat{x},\hat{y})}$$
if $C_{L_k}(x)$ is a $\mu$ good box and $\hat{y}\in \partial^{ext}C_{L_k}(x)$, and 
$$G'(\hat{x},\hat{y})= \sum_{\substack{\hat{z}_1\in \partial^{in} C_{L_{2k}}(x) \\ \hat{z}_2\in \partial^{ext} C_{L_{2k}}(x) \\ \hat{z}_3\in \partial^{in} C_{L_{k}}(z_2)  }} e^{\sqrt{L_k}} \frac{C (2L_k)^{d/2}}{1+|k_{\hat{x}}-k_{\hat{z}_1}|}|g\hat{\Delta}(\hat{z_1},\hat{z_2})| P(\hat{z}_2)e^{-d_G(\hat{z_2},\hat{z_3})} |g\hat{\Delta}(\hat{z_3},\hat{y}) | P  (\hat{y})$$
if $C_L(x)$ is a bad box. 
%

Thanks to Proposition \ref{Pi MSA technical}, there is a constant $C$ independent of $L_k$ such that for the second case :
$\|G'\|_{L^1}\leq C' e^{2\sqrt{L_k}}e^{-\mu L_k} $. We can then recover the usual tools, using that $e^{-\mu L_k}$ dominate the other terms for $L_k$ large. In particular because for any path from $x$ to $\partial C_L(x)$ there is at least $(\frac{L_{k+1}}{L_k}-3)$ $\mu$ good boxes. So, with the same argument as in the proof of Proposition \ref{PSommeDist},
\begin{equation*}
\sum_{\hat{y}\in \partial^{in} C_L(x)}e^{-d_{G'}(\hat{x},\hat{y})} \leq e^{-\mu (L_{k+1}-3L_k)-\ln(1-\|G'\|_{\ell^1 max})}	
\end{equation*}
\end{proof}

\section{Proof of the corollaries}\label{section: dynamical localization}
As said, Corollaries \ref{TAbsenceDiff} and \ref{CUnitaire} do not follow logically from Theorem \ref{Tlocalisation}; instead one should go trough the MSA once again and refine several estimates. 
This work has been carried over in \cite{Damanik-Stollmann}, and one indicates here only the main steps as well as the few needed extra adaptations. 

Let us start with Corollary \ref{TAbsenceDiff}. 
 

\begin{Prop}
there exist $p>0$ (and one can take $p\to \infty$ as $\epsilon \to 0$) such that: 
\begin{equation}
\mathbb{E}(\sup_{t>0} \sum_{x\in\Z^d }\sum_k |x|^p |\check{\phi}(x,k,t)|^2)<\infty
\end{equation}
for any $\check{\phi}(x,k,0)$ defined on a bounded support. 
\end{Prop}
\begin{proof}
Thanks to the MSA carried over in this paper, one can check that the results of \cite{Damanik-Stollmann} holds; in particular the assumptions of Theorem 3.1 in \cite{Damanik-Stollmann} are satisfied.
\end{proof}

In order to recover $\phi$ from $\check{\phi}$ we use the following proposition.
Remind that, thanks to \eqref{LtwoBound}, we have $\|H(t)\|_{L^1[0;T]}\leq \sqrt{T} \|H(t)\|_{L^2[0;T]}$.


\begin{Prop}
Let $\psi(t) \in L^2 (\Z^d)$ satisfying $\|\psi(t)\|_{L^2}=1$ for all $t\in \R$ be a solution of 
\begin{equation}
i\partial_t \psi(t) = A(t)\psi(t)
\end{equation}
where for any $t$ $A(t)$ is hermitian, $C=\|A(.)\|_{L^1([0,T])}<\infty$ and $(x,A(t)y)=0$ if $|x-y|>1$.
For any $t\in[0,T]$ and any $x_0\in\mathbb{Z}^d$, we have
\begin{equation}
\sum_{|z-x_0|<R}|\psi(z,t)|^2 \geq |\psi(x_0,0)|^2\Big(1-e^{ C}\sum_{k \geq R} \frac{(2dC)^k}{k!}\Big)-e^{C}\sum_{k \geq R} \frac{(2dC)^k}{k!}|\psi(x_0)|.
\end{equation}
\end{Prop}

\begin{proof}
Let's separate $\psi(0)=\mathds{1}_{x=x_0}\psi(0)+\mathds{1}_{x\neq x_0}\psi(0)$. Because the $A(t)$ is hermitian, there exists $U(t)$ unitary such that 
\begin{equation}
\psi(t)=U(t)\psi(0) = U(t)(\mathds{1}_{x=x_0}\psi(0))+U(t)(\mathds{1}_{x\neq x_0}\psi(0))
\end{equation}  
Calling $\psi_1 = U(t)(\mathds{1}_{x=x_0}\psi(0)),\psi_2=U(t)(\mathds{1}_{x\neq x_0}\psi(0))$  we have $(\psi_1,\psi_2)=0$ and $\|\psi_1\|^2+\|\psi_2\|^2=1$. Because $\mathds{1}_{|z-x_0|<R}$ is a projector, 
\begin{align*}
\big(\psi_1+\psi_2,\mathds{1}_{|z-x_0|<R}(\psi_1+\psi_2)\big) & = \big(\psi_1,\mathds{1}_{|z-x_0|<R}\psi_1\big)+\big(\psi_2,\mathds{1}_{|z-x_0|<R}\psi_2\big)+2\big(\psi_1,\mathds{1}_{|z-x_0|<R}\psi_2\big) \\
& \geq \big(\psi_1,\mathds{1}_{|z-x_0|<R}\psi_1\big)-2|\big(\psi_2,\mathds{1}_{|z-x_0|\geq R}\psi_1\big)| \\
& \geq \|\psi_1\|^2- \| \mathds{1}_{|z-x_0|\geq R}\psi_1\|^2- 2 |(\psi_2,\mathds{1}_{|z-x_0|\geq R}\psi_1)| \\
& \geq \|\psi_1\|^2- \| \mathds{1}_{|z-x_0|\geq R}\psi_1\|^2- 2\|\mathds{1}_{|z-x_0|\geq R}\psi_1 \|
\end{align*}
We now proof that the locality of $A(t)$ implies that  $\| \mathds{1}_{|z-x_0|\geq R}\psi_1\|^2$ is small. 

\begin{equation*}
i\frac{d}{dt}\psi_1(y,t)  = A(t)\psi_1(y,t)  = \sum_{|y'-y|\leq 1}A_{y,y'}(t)\psi_1(y',t).
\end{equation*}
Hence
\begin{equation*}
\frac{d}{dt}|\psi_1(y,t)|  \leq \sum_{|y'-y|\leq 1}|A_{y,y'}(t)| |\psi_1(y',t)| \leq \|A(t)\|\sum_{|y'-y|\leq 1} |\psi_1(y',t)| \\
\end{equation*}
Let now $a(y,t)$ solution of the system 
\begin{equation}
\begin{cases}
\frac{d}{dt}a(y,t) = \|A(t)\|\sum_{|y'-y|\leq 1} a(y',t) \\
a(y,0)=|\psi_1(x_0,0)|\mathds{1}_{y=x_0}
\end{cases}
\end{equation}
We have then for any $(y,t)$
\begin{equation}
|\psi_1(y,t)|\leq a(y,t)
\end{equation}
We can evaluate $a$ with the following remark : 
Let $X(t)$ be the classical markovian random walk on $\mathbb{Z}$ of variable rate $\|A(t)\|$ and starting at point $x_0$. Its generator is 
\begin{equation}
\frac{d}{dt}\mathbb{P}_{x_0}(X(t)=y) = \|A(t)\| \sum_{|y'-y|}(\mathbb{P}_{x_0}(X(t)=y')-\mathbb{P}_{x_0}(X(t)=y))
\end{equation}
and then we have 
\begin{equation}
e^{-(2d+1)\int_0^t\|A(u)\|du}a(y,t)=a(x_0,0)\mathbb{P}_{x_0}(X(t)=y)
\end{equation}
We can then deduce 
\begin{equation}
\sum_{y\geq R}a(y,t)\leq a(x_0,0)e^{(2d+1)\int_0^t\|A(u)\|du}\mathbb{P}(N_{2d\int_0^t\|A(u)\|du}\geq R)
\end{equation}
where $N_{2d\int_0^t\|A(u)\|du}$ is the Poisson process of parameter $2d\int_0^t\|A(u)\|du$. So for any $t\leq T$
\begin{equation}
\sum_{y\geq R}a(y,t)\leq a(x_0,0)e^{C}\sum_{k \geq R} \frac{(2dC)^k}{k!}
\end{equation}
We can now conclude 
\begin{align*}
\sum_{|z-x_0|<R}|\psi(z,t)|^2 & = \big(\psi_1+\psi_2,\mathds{1}_{|z-x_0|<R}(\psi_1+\psi_2)\big) \\
 & \geq \|\psi_1\|^2- \| \mathds{1}_{|z-x_0|\geq R}\psi_1\|^2- \|\mathds{1}_{|z-x_0|\geq R}\psi_1 \| \\
 & \geq |\psi(x_0,0)|^2-|\psi(x_0,0)|^2(e^{C}\sum_{k \geq R} \frac{(2dC)^k}{k!})^2-|\psi(x_0,0)|(e^{C}\sum_{k \geq R} \frac{(2dC)^k}{k!})
\end{align*}

\end{proof}

The above proposition and the dynamical localisation of $\check{\phi}$ enable us to conclude:
\begin{Prop}
For any $\epsilon>0$, there exist some constants $C_\epsilon$, $D_\epsilon$ such that
\begin{equation}
C_\epsilon \sum_{x\in\Z^d }\sum_k |x|^p |\check{\phi}(x,k,t)|^2  + D_\epsilon \geq \sum_{x_0\in\Z^d } |x_0|^{p-\epsilon} |\phi(x_0,t)|^2
\end{equation}
\end{Prop}

\begin{proof}
Let $\epsilon>0$. Let now $x_0 \mapsto R(x_0)$ be such that 
\begin{equation}
\sum_{x_0\in \Z^d} |x_0|^p\sum_{k \geq R(x_0)} \frac{(2dC)^k}{k!}<\infty
\end{equation} 
and such that, for all $x_0 \in \Z^d$, 
\begin{equation}
 e^{C}\sum_{k \geq R(x_0)} \frac{(2dC)^k}{k!} < \frac{1}{2}
\end{equation}
moreover that $|x-x_0|<R(x_0)$ then $|x-x_0|<(1+\epsilon)R(x)$,
and such there is constant $C_\epsilon$ such that  
\begin{equation}
\sum_{|x-x_0|\leq (1+\epsilon)R(x)} |x_0|^{p-\epsilon} \leq C_{\epsilon}|x|^{p}
\end{equation} 
for $|x_0|>1$. For example we could have chosen $R(x)=ln(x)^2$ for large $x$. 
\begin{align*}
& \mathbb{E}(\sup_{t>0} \sum_{x\in\Z^d }\sum_k |x|^p |\check{\phi}(x,k,t)|^2 ) \\  \qquad & = \mathbb{E}(\sup_{t>0} \sum_{x\in\Z^d } |x|^p \frac{1}{T}\int_t^{t+T} |\phi(x,u)|^2 du ) \\
\qquad & \geq \frac{1}{C_{\epsilon}} \mathbb{E}(\sup_{t>0} \sum_{x\in\Z^d } \sum_{|x-x_0| \leq (1+\epsilon)R(x)}\frac{1}{T}\int_t^{t+T} |x_0|^{p-\epsilon} |\phi(x,u)|^2 du ) \\
\qquad & \geq \frac{1}{C_{\epsilon}} \mathbb{E}(\sup_{t>0} \sum_{x_0\in\Z^d } |x_0|^{p-\epsilon} \frac{1}{T}\int_t^{t+T} \sum_{|x-x_0|\leq R(x_0)}|\phi(x,u)|^2 du ) \\
\qquad & \geq \frac{1}{C_{\epsilon}} \mathbb{E}(\sup_{t>0} \sum_{x_0\in\Z^d } |x_0|^{p-\epsilon} \frac{1}{T}\int_t^{t+T} |\psi(x_0,t)|^2-|\psi(x_0,t)|^2(e^{C}\sum_{k \geq R} \frac{(2dC)^k}{k!})^2
\\ & \qquad-|\psi(x_0,t)|(e^{C}\sum_{k \geq R} \frac{(2dC)^k}{k!}) du  \\
\qquad & \geq \frac{1}{2 C_{\epsilon}} \mathbb{E}\big(\sup_{t>0} \sum_{x_0\in\Z^d } |x_0|^{p-\epsilon} |\psi(x_0,t)|^2\big) - e^C\frac{1}{C_{\epsilon}}
\sum_{x_0\in \Z^d} |x_0|^p\sum_{k \geq R(x_0)} \frac{(2dC)^k}{k!}
\end{align*}
So 
\begin{equation}
\mathbb{E}(\sup_{t>0} \sum_{x_0\in\Z^d } |x_0|^{p-\epsilon} |\psi(x_0,t)|^2)<\infty
\end{equation}
\end{proof}

Let us now come to Corollary \ref{CUnitaire}: 
\begin{proof}[Proof of Corollary \ref{CUnitaire}.]
Since
$$\psi_{\overline\lambda} (\cdot ,0) = \sum_{k\in \Z} \hat{\psi} (\cdot, k),$$
we can write 
$$H_{eff} (x,y) = \sum_{(k,l)\in \Z^2} \sum_{\bar\lambda \in [0,\nu[} \bar\lambda \psi_{\bar\lambda} (x,k) \bar\psi_{\bar\lambda}(y,l) = \sum_{k,l} \big( (x,k),\eta(\hat{H}) (y,k)\big) $$ 
with 
$$\eta:\R \to \R , s \mapsto  \eta(s) = 1_{[0,\nu[} (s) s.  $$
Again, thanks to the MSA shown in this paper, and the deterministic exponential decay along the frequency axis of the eigenfunctions under Assumption (C1), 
we can reuse the methods leading to Theorem 3.1 in \cite{Damanik-Stollmann}, to get our result. 
\end{proof}

\bibliographystyle{plain}
\bibliography{BiblioLocalisationAnderson.bib}

\end{document}